\documentclass[12pt]{article}
\usepackage{bbm}
\usepackage{float}
\usepackage{amsmath}
\usepackage{graphicx}
\usepackage{natbib}
\usepackage{amssymb}
\usepackage{amsthm}
\usepackage{mathrsfs}
\usepackage{subfigure}
\usepackage{amsfonts} 
\usepackage{docmute}
\usepackage{booktabs}
\usepackage{listings}
\usepackage[autostyle]{csquotes}
\usepackage{hyperref}
\usepackage{url} 
\usepackage{xcolor}
\usepackage{comment}
\usepackage{algorithm}
\usepackage{algorithmic}
\usepackage{spverbatim}
\newcommand{\<}{\langle}
\renewcommand{\>}{\rangle}
\graphicspath{ {./images/} }

\newtheorem{lemma}{Lemma}
\newtheorem{assumption}{Assumption}
\newtheorem{remark}{Remark}
\newtheorem{proposition}{Proposition}
\providecommand{\norm}[1]{\lVert#1\rVert}
\DeclareMathOperator{\essinf}{ess\,inf}
\DeclareMathOperator{\esssup}{ess\,sup}

\newcommand{\blind}{0}

\begin{document}

\def\spacingset#1{\renewcommand{\baselinestretch}%
{#1}\small\normalsize} \spacingset{1}


\if0\blind
{
  \title{\bf Covariance--based rational approximations of fractional SPDEs for computationally efficient Bayesian inference}
  \author{David Bolin, Alexandre B. Simas, and Zhen Xiong \hspace{.2cm}\\
    Computer, Electrical and Mathematical Sciences and Engineering \\ Division, 
    King Abdullah University of Science and Technology\\
    Thuwal 23955-6900, Saudi Arabia}
  \maketitle
} \fi

\if1\blind
{
  \bigskip
  \bigskip
  \bigskip
  \begin{center}
    {\LARGE\bf Title}
\end{center}
  \medskip
} \fi

\bigskip
\begin{abstract}
The stochastic partial differential equation (SPDE) approach is widely used for modeling large spatial datasets. It is based on representing a Gaussian random field $u$ on $\mathbb{R}^d$ as the solution of an elliptic SPDE $L^\beta u = \mathcal{W}$ where $L$ is a second-order differential operator, $2\beta \in \mathbbm{N}$ is a positive parameter that controls the smoothness of $u$ and $\mathcal{W}$ is  Gaussian white noise. A few approaches have been suggested in the literature to extend the approach to allow for any smoothness parameter satisfying $\beta>d/4$. Even though those approaches work well for simulating SPDEs with general smoothness, they are less suitable for Bayesian inference since they do not provide approximations which are Gaussian Markov random fields (GMRFs) as in the original SPDE approach. We address this issue by proposing a new method based on approximating the covariance operator $L^{-2\beta}$ of the Gaussian field $u$ by a finite element method combined with a rational approximation of the fractional power. This results in a numerically stable GMRF approximation which can be combined with the integrated nested Laplace approximation (INLA) method for fast Bayesian inference. A rigorous convergence analysis of the method is performed and the accuracy of the method is investigated with simulated data. Finally, we illustrate the approach and corresponding implementation in the R package rSPDE via an application to precipitation data which is analyzed by combining the rSPDE package with the R-INLA software for full Bayesian inference. 
\end{abstract}

\noindent%
\vfill

\newpage
\spacingset{1} 

\section{Introduction}
\label{sec:intro}

Handling many observations from a Gaussian random field in spatial statistics can be challenging since the related computational tasks involve factorizations of large covariance matrices which are usually dense. This is often referred to as the ``big $N$ problem'' \citep{banerjee2003hierarchical}, and various approaches have been suggested to handle the computational issues \citep[see, e.g.,][for a recent comparison]{heaton2019}.
One of the most widely used methods is the SPDE approach by \cite{lindgren2011}. This is based on the fact that a centered Gaussian random field $u$ on the spatial domain $\mathcal{D} = \mathbb{R}^d$ with an isotropic Mat\'ern covariance function \citep{matern60}, 
\begin{equation}\label{eq:matern}
	\varrho(\boldsymbol{s}) = \frac{\sigma^2}{2^{\nu-1}\Gamma(\nu)}(\kappa \norm{\boldsymbol{s}})^\nu K_\nu(\kappa \norm{\boldsymbol{s}}),\quad \boldsymbol{s}\in\mathbb{R}^d,
\end{equation}
can be represented as a solution to the stochastic partial differential equation (SPDE) 
\begin{equation}\label{eq:spde_intro}
	(\kappa^2-\Delta)^{\beta} (\tau u) = \mathcal{W} \quad \text{ in $\mathcal{D}$}.
\end{equation}
Here, $\Gamma(\cdot)$ is the Gamma function, $K_\nu$ is a
modified Bessel function of the second kind, $\Delta$ is the Laplace operator, and $\mathcal{W}$ is Gaussian white noise. The parameter  $\kappa>0$ controls the practical correlation range, $\sigma^2$ is the variance,  $\tau^2 = \Gamma(\nu) / (\sigma^2 \kappa^{2\nu} (4\pi)^{d/2} \Gamma(\nu + d/2))$, and the fractional power $\beta$ is related to the smoothness parameter $\nu>0$ via the relation $2\beta = \nu + d/2$ \citep{whittle1963stochastic}. \cite{lindgren2011} used this representation to construct computationally efficient Gaussian Markov Random Field (GMRF) approximations of Gaussian Mat\'ern fields by considering the SPDE on a bounded domain $\mathcal{D}$, restricting the smoothness to $2\beta\in\mathbb{N}$, and then performing a finite element method (FEM) discretization. 

The SPDE approach has become widely used in applications, and has initiated a great number of extensions and generalizations  \citep{lindgren2022spde}. The reason for this is not only the computational benefits, but also that it provides a flexible framework for defining more sophisticated models for spatial data. It, for example, facilitates the construction of non-stationary Gaussian random fields by allowing the parameters $\kappa$ and $\tau$ to be spatially varying \citep{lindgren2011,fuglstad2015does}, and allows for the construction of Mat\'ern-like random fields on more general manifolds by defining such fields via the SPDE \eqref{eq:spde_intro} posed on the manifold \citep{lindgren2011,bolin11}. 

One of the main criticisms of the SPDE approach is the requirement $2\beta \in \mathbbm{N}$, which restricts the possible values of the corresponding smoothness parameter $\nu$ of the Mat\'ern covariance function. Given the importance of $\nu$ when performing prediction, as shown by \citet{stein1999interpolation} and \citet{bk2022measure}, several methods for removing the restriction of $2\beta \in \mathbbm{N}$ have been proposed. \citet[][Author's response]{lindgren2011}  proposed to construct a GMRF approximation by approximating the spectrum of a Gaussian Mat\'ern field by a spectrum that is a reciprocal of a polynomial. This method is applicable for stationary models but it can not be applied to non-stationary models, and it has a fixed accuracy which may not be sufficient for certain applications. \citet{bolin2020numerical} proposed combining the FEM approximation of \citet{lindgren2011} with a quadrature approximation of the fractional operator to obtain a numerical method that works for any $\beta>d/4$ and can be made arbitrarily accurate. That work also provided a theoretical convergence analysis of the method, which was extended in \citet{bolin2018weak} and \citet{herrmann2020multilevel}.  \citet{rational_spde} later proposed a different type of approximation referred to as the rational SPDE approach, which has a lower computational cost. 

Even though the methods that work for non-stationary models with general smoothness are computationally efficient, they are much less used than the standard SPDE approach for statistical applications. The reason for this is that non-fractional SPDE models work in combination with the integrated nested Laplace approximation (INLA) method \citep{rue09} and are implemented in the widely used \texttt{R-INLA}  \citep{lindgren2015software}  \texttt{R} \citep{Rsoftware} package.
This software facilitates including SPDE-based models in general Bayesian latent Gaussian models, and the great majority of all applications of the SPDE approach have been done via this software. 

Unfortunately, the methods of \citet{bolin2020numerical} and \citet{rational_spde} provide approximations which are not compatible with \texttt{R-INLA}.
The reason is that the methods do not yield a Markov approximation, so the precision matrix obtained from the approximations are not sparse.  
The covariance matrix of the approximations are of the form $\mathbf{P}\mathbf{Q}^{-1}\mathbf{P}$, where both $\mathbf{P}$ and $\mathbf{Q}$ are sparse matrices that depend on the parameters of the model. 
To achieve a sparse precision matrix, which is necessary for \texttt{R-INLA}, \citet{rational_spde} showed that one can work with a latent model with sparse precision matrix $\mathbf{Q}$ if the projection matrix $\mathbf{A}$, which connects the locations of the mesh for the FEM approximation and the observation locations (see Section \ref{sec:overview} for details), is adjusted to $\widehat{\mathbf{A}} = \mathbf{A}\mathbf{P}$. This matrix, however, depends on the model parameters, which is not allowed in \texttt{R-INLA}.

The main goal of this work is to solve this problem by proposing an alternative rational approximation. 
The main idea is to approximate the covariance operator of the random field directly, instead of first approximating the solution $u$ and then deriving the corresponding covariance operator. 
This provides an approximation suitable for \texttt{R-INLA}, which is more numerically stable than that of the original rational SPDE approach. 
The proposed method is implemented in the R package \texttt{rSPDE} \citep{rspde}, which is available on CRAN and has an interface to \texttt{R-INLA}. 
Using the package, we show that the proposed method facilitates full Bayesian inference of all model parameters, including $\beta$, for latent Gaussian models based on fractional SPDEs.

The outline of the paper is as follows. 
In Section~\ref{sec:overview}, we give an overview of the model structure of the proposed approximation and show how it can be used for computationally efficient inference. The mathematical details and justifications of the method are provided in Sections~\ref{sec:prelim} to \ref{GMRF}.
Specifically, in Section~\ref{sec:prelim}, we introduce the generalized Whittle--Mat\'ern fields, which contain most of the previously proposed non-stationary SPDE-based Gaussian random fields as special cases, and for which our proposed method is applicable. 
In that section, we also provide the details of the FEM approximations. 
The new covariance-based rational approximation is introduced in Section~\ref{rationalapprox}, where we also prove that it provides an approximation of the covariance function of the generalized Whittle--Mat\'ern field with an explicit rate of convergence in the $L_2$-norm. 
In Section~\ref{GMRF}, we show that the covariance-based rational approximation can be represented as a GMRF, and illustrate how this can be used for statistical inference.  
Some of the details of the \texttt{rSPDE} implementation are discussed in Section~\ref{package}, and a comparison in terms of the accuracy of approximating covariance function by our method and some other methods is provided in Section~\ref{sec:numexp}. 
An application to modeling of precipitation data is presented in Section~\ref{application} and the article concludes with a discussion in Section~\ref{discussion}. Finally, the supplementary materials contain seven appendices which provide further technical details and proofs.

\section{Overview of the approximation strategy}\label{sec:overview}
As mentioned in the introduction, the main idea behind our strategy is to directly approximate the covariance operator of the random field. In this section we show the structure of the resulting approximation and also provide an illustration on how it can be used for inference in a simple application. More details will be given in later sections.
The covariance-based rational approximation of the Whittle--Mat\'ern field $u(\boldsymbol{s})$ defined in \eqref{eq:spde_intro}, whose covariance operator is $(\kappa^2 -\Delta)^{-2\beta}$, uses a combination of the finite element method and rational approximations in order to approximate $u(\boldsymbol{s})$ as $u_n(\boldsymbol{s})=\sum_{j=1}^{n} w_j \varphi_j(\boldsymbol{s})$, where $\{w_j\}_{j = 1}^{n}$ are stochastic weights and $\{\varphi_j\}_{j = 1}^{n}$ are FEM basis functions. We denote $\boldsymbol{w}=[w_1,...,w_{n}]^\top$. With our approximation, $\boldsymbol{w}$ can be expressed as a sum of $m+1$ independent GMRFs with sparse precision matrices:
\begin{equation}\label{stoc_weights_sum_representaion}
	\boldsymbol{w}=\sum_{i=1}^{m+1}\boldsymbol{x}_i,\quad\hbox{where}\quad
	\boldsymbol{x}_i \sim N(\boldsymbol{0},\boldsymbol{Q}_i^{-1}),\quad
	\boldsymbol{x}_i=\begin{pmatrix}x_{i1} &\cdots& x_{in}\end{pmatrix}^\top.
\end{equation}

Any linear predictor in \texttt{R-INLA} has this form, which means that we can perform inference in a computationally efficient manner based on the covariance-based rational approximation by using the same ideas as are used in  \texttt{R-INLA}.
For example, suppose that we observe $y_1,\ldots,y_N$, $N\in\mathbb{N}$, where
\begin{equation}\label{hiermodel}
	y_i=u(\boldsymbol{s}_i)+\epsilon_i,\quad i=1,...,N, 
\end{equation}
$\boldsymbol{s}_1,\ldots, \boldsymbol{s}_N \in\mathbb{R}^d$ are spatial locations, and $\pmb{\epsilon}=[\epsilon_1,...,\epsilon_N]^\top \sim N(\textbf{0},
\boldsymbol{Q}^{-1}_{\pmb{\epsilon}})$ for some sparse matrix $\boldsymbol{Q}_{\pmb{\epsilon}}$, such as 
 $\boldsymbol{Q}_{\pmb{\epsilon}} = \frac{1}{\sigma_{\epsilon}^2} \boldsymbol{I}_N$ if we have independent measurement noise. 
Defining $\boldsymbol{y}=[y_1,...,y_N]^\top$,
\eqref{hiermodel} can be written in matrix form as 
$\boldsymbol{y}=\textbf{A}\boldsymbol{w}+\pmb{\epsilon}$, 
where $\textbf{A}$ is the projector matrix with elements $A_{ij} = \varphi_j(\boldsymbol{s}_i)$.
Let $\boldsymbol{X}=[\boldsymbol{x}^\top_1,...,\boldsymbol{x}^\top_{m+1}]^\top$. Then, the precision matrix of $\boldsymbol{X}$ is the block diagonal matrix
\begin{equation}
	\label{Q}
	\boldsymbol{Q}=\text{diag}(\boldsymbol{Q}_1,\ldots, \boldsymbol{Q}_{m+1}).
\end{equation}
Writing the model in terms of the weights $\boldsymbol{X}$ allows us to equivalently write the model as $\boldsymbol{y}=\overline{\boldsymbol{A}}\boldsymbol{X}+\pmb{\epsilon}$ where $\overline{\boldsymbol{A}}$ is a block matrix of size $N\times n(m+1)$ obtained by combining $m+1$ copies of $\boldsymbol{A}$ as  $\overline{\boldsymbol{A}}=\begin{bmatrix}\boldsymbol{A}&\cdots&\boldsymbol{A}\end{bmatrix}$.
Thus, 
$\boldsymbol{y}\vert\boldsymbol{X} \sim N(\overline{\boldsymbol{A}}\boldsymbol{X},\boldsymbol{Q}^{-1}_{\pmb{\epsilon}})$
and
$\boldsymbol{X} \sim N(\textbf{0},\boldsymbol{Q}^{-1})$,
where $\boldsymbol{Q}$ is given in \eqref{Q}. Standard results for latent Gaussian models then give us that the posterior distribution of $\boldsymbol{X}$ is $\boldsymbol{X}\vert\boldsymbol{y} \sim N(\pmb{\mu}_{\boldsymbol{X}\vert\boldsymbol{y}},\boldsymbol{Q}^{-1}_{\boldsymbol{X}\vert\boldsymbol{y}})$,
where
\begin{equation}
	\label{pos_latentF}
	\pmb{\mu}_{\boldsymbol{X}\vert\boldsymbol{y}}=\boldsymbol{Q}^{-1}_{\boldsymbol{X}\vert\boldsymbol{y}}\overline{\boldsymbol{A}}^\top\boldsymbol{Q}_{\pmb{\epsilon}}\boldsymbol{y}\quad \hbox{and}
	\quad
	\boldsymbol{Q}_{\boldsymbol{X}\vert\boldsymbol{y}}=\overline{\boldsymbol{A}}^\top\boldsymbol{Q}_{\pmb{\epsilon}}\overline{\boldsymbol{A}}+\boldsymbol{Q}.
\end{equation}
Finally, we can obtain the marginal likelihood, $\ell(\boldsymbol{y})$, of $\boldsymbol{y}$ as
\begin{equation}
	\label{marginal_likelihood}
	\begin{aligned}
		2\ell(\boldsymbol{y}) = &\log{|\boldsymbol{Q}|}+\log{|\boldsymbol{Q}_{\pmb{\epsilon}}|}-\log{|\boldsymbol{Q}_{\boldsymbol{X}\vert\boldsymbol{y}}|}-\pmb{\mu}_{\boldsymbol{X}\vert\boldsymbol{y}}^\top\boldsymbol{Q}\pmb{\mu}_{\boldsymbol{X}\vert\boldsymbol{y}} \\
		& -(\boldsymbol{y}-\overline{\boldsymbol{A}}\pmb{\mu}_{\boldsymbol{X}\vert\boldsymbol{y}})^\top\boldsymbol{Q}_{\pmb{\epsilon}}(\boldsymbol{y}-\overline{\boldsymbol{A}}\pmb{\mu}_{\boldsymbol{X}\vert\boldsymbol{y}})-  n\log(2\pi).
	\end{aligned}
\end{equation}

The sparsity of $\boldsymbol{Q}$ is essential for computation. For instance, evaluating $\log|\boldsymbol{Q}|$ in the likelihood can be done efficiently based on sparse Cholesky decomposition \citep{rue2005gaussian}. Sparsity of $\boldsymbol{Q}$ also facilitates computationally efficient sampling of $\boldsymbol{w}$, and hence of $u$. See Appendix~\ref{sec:algorithm} for further details on the methods for sampling and likelihood evaluation.

\section{Whittle--Mat\'ern fields and FEM approximation}
\label{sec:prelim}
In this section we introduce the class of fractional-order SPDEs we are interested in as well as their FEM approximations. The model assumptions are presented in Section \ref{sec:continuous_model}, and in Section \ref{finite_element} we introduce the FEM approximations and study their convergence.

Let us begin by introducing some notation that will be needed later on. Given a bounded domain $\mathcal{D}\subset \mathbb{R}^d$, $d \in \{1,2,3\}$, we denote by $L_2(\mathcal{D})$ the
Lebesgue space of square-integrable real-valued functions endowed with the inner product ${(\phi,\psi)_{L_2(\mathcal{D})} = \int_\mathcal{D} \phi(\boldsymbol{x})\psi(\boldsymbol{x}) d\boldsymbol{x}}$. We denote the Sobolev space of order $k$ by $H^k(\mathcal{D})$:
$$H^k(\mathcal{D}) = \left\{w\in L_2(\mathcal{D}): D^\gamma w \in L_2(\mathcal{D}), \forall \gamma\in\mathbb{N}^d, |\gamma|\leq k\right\},$$
where we are using the multiindex notation for the differential operator $D^\gamma$, and $(\cdot,\cdot)_{H^k(\mathcal{D})}$ is the Sobolev inner product:
$$(u,v)_{H^k(\mathcal{D})} = \sum_{\gamma\in\mathbb{N}^d: |\gamma|\leq k} (D^\gamma u, D^\gamma v)_{L_2(\mathcal{D})}.$$ We denote by $H_0^1(\mathcal{D})$ the closure of $C_c^\infty(\mathcal{D})$ in $H^1(\mathcal{D})$, where $C_c^\infty(\mathcal{D})$ is the set of infinitely differentiable functions with compact support on $\mathcal{D}$. Additional notations needed for the theoretical analysis are given in Appendix~\ref{notation}. 

\subsection{Model assumptions}
\label{sec:continuous_model}
We are interested in the class of Gaussian random fields on $\mathcal{D}$ that can be represented as solutions to SPDEs of the form
\begin{equation}
	\label{general_main_equation_continuous}
	L^{\beta}(\tau u) = \mathcal{W} \quad  \text{in}\ \mathcal{D},
\end{equation}
where $L^\beta$ is a fractional power (in the spectral sense) of a second-order elliptic differential operator $L$ which determines the covariance structure of $u$, $\tau > 0 $ is a constant parameter, and $\mathcal{W}$ is Gaussian white noise on $L_2(\mathcal{D})$.
We have the following assumptions on $\mathcal{D}$:
\begin{assumption}\label{assumpDomain}
	The domain $\mathcal{D}$ is an open, bounded, convex polytope with closure $\overline{\mathcal{D}}$. 
\end{assumption}

Under Assumption \ref{assumpDomain}, we may define $H^2_{\mathcal{N}}(\mathcal{D}) = \{w\in H^2(\mathcal{D}): \partial w/\partial \nu = 0\hbox{ on }\partial\mathcal{D}\}$, where $\nu$ is the outward unit normal vector to $\partial\mathcal{D}$. Indeed, the expression $\partial w/\partial\nu = 0$ on $\partial\mathcal{D}$ makes sense since the trace of $Dw$ is well-defined in this case \citep[see, e.g.,][Theorem 4.6]{evansfineproperties}.
Let us now describe the assumptions on the differential operator $L$:
\begin{assumption}\label{assumpOperator}
	The operator $L$ is given in divergence form by ${Lu = -\nabla \cdot (\boldsymbol{H}\nabla u) + \kappa^2 u}$, and is equipped either with homogeneous Dirichlet or Neumann boundary conditions. 
	Furthermore, the function $\boldsymbol{H}: \overline{\mathcal{D}} \to \mathbbm{R}^{d \times d}$ is symmetric, Lipschitz continuous and uniformly positive definite, and $\kappa: \mathcal{D} \to \mathbbm{R}$ is an essentially bounded function, that is, 
	$$\esssup_{x\in\mathcal{D}} \kappa(x) = \inf\{a\in \mathbb{R}:\lambda(\{x:\kappa(x)>a\}) = 0\} <\infty.$$  
	Under Neumann boundary conditions, we additionally require that 
	$$\essinf_{x\in\mathcal{D}} \kappa(x) = \sup\{b\in \mathbb{R}:\lambda(\{x:\kappa(x)<b\}) = 0\} \geq \kappa_0>0,$$ 
	where $\lambda$ is the Lebesgue measure on $\mathcal{D}$.
\end{assumption}  

The SPDE \eqref{general_main_equation_continuous} under Assumptions \ref{assumpDomain} and \ref{assumpOperator} defines a class of models that have previously been considered by \citet{bolin2020numerical,cox_Kirchner,herrmann2020multilevel,rational_spde} and is referred to as generalized Whittle--Mat\'ern fields. It contains many previously proposed non-stationary SPDE-based spatial Gaussian random field models as special cases, such as those by \citet{lindgren2011,fuglstad2015does,fuglstad2019constructing,Hildeman2020}, and the method that we later introduce thus also applies to those models and their fractional extensions. 

In the case of Dirichlet boundary conditions, define the space $V = H^1_0(\mathcal{D}) \subset L_2(\mathcal{D})$, and in the case of Neumann boundary conditions let $V = H^1(\mathcal{D}) \subset L_2(\mathcal{D})$. Then, under Assumptions \ref{assumpDomain} and \ref{assumpOperator}, $L$ induces the following continuous and coercive bilinear form on $V$:
\begin{equation}
	\label{bilinear_form}
	a_L(v, u) = (\boldsymbol{H}\nabla u, \nabla v )_{L_2(\mathcal{D})} + (\kappa^2 u, v )_{L_2(\mathcal{D})}, \quad u,v \in V.
\end{equation}
\begin{remark}\label{rem:H2regul}
	Under Assumptions \ref{assumpDomain} and \ref{assumpOperator}, if $f\in L_2(\mathcal{D})$, then there exists a unique solution $u$ of $Lu = f$ and the operator $L$ is $H^2(\mathcal{D})$-regular, that is, $u\in H^2(\mathcal{D})\cap H^1_0(\mathcal{D})$ under Dirichlet boundary conditions, whereas under Neumann boundary conditions, we have $u\in H^2_\mathcal{N}(\mathcal{D})$. See, for instance, \cite[Theorem 3.2.1.2]{grisvard} for Dirichlet boundary conditions or \cite[Theorem 3.2.1.3]{grisvard} for Neumann boundary conditions.
\end{remark}

By remark \ref{rem:H2regul}, specifically by the existence and uniqueness of the solution to the equation $Lu = f$, we can define the inverse operator $L^{-1}: L_2(\mathcal{D})\to L_2(\mathcal{D})$. 
By Rellich-Kondrachov theorem \cite[Theorem 4.11]{evansfineproperties}, $L^{-1}$ is a compact operator, and observe that $L^{-1}$ is self-adjoint, see Appendix~\ref{notation} for a justification. 
Hence, by the spectral theorem for self-adjoint and compact operators, there exists an orthonormal basis $\{e_j\}_{j\in\mathbb{N}}$ in $L_2(\mathcal{D})$ formed by eigenvectors of $L$ whose eigenvalues $\{\lambda_j\}_{j\in\mathbb{N}}$ are non-negative and can be arranged in a non-decreasing order.

\begin{remark}\label{rem:WeylsLaw}
	Under Assumptions \ref{assumpDomain} and \ref{assumpOperator}, 
	the operator $L$ satisfies the Weyl's law,
	that is, there exist $c,C>0$ such that for every $j\in\mathbb{N}$,
	$c j^{2/d} \leq \lambda_j \leq C j^{2/d}.$
	See \citet[Theorem 6.3.1]{davies1996spectral} for the Dirichlet case. 
	For the Neumann case, the Weyl's law holds for the case
	in which $\boldsymbol{H}$ is a constant diagonal matrix(\cite{fedosov1,fedosov2}), in particular, it holds for the Neumann Laplacian. 
	The result for a general $\boldsymbol{H}$ satisfying Assumption
	\ref{assumpOperator} is a direct consequence of the
	Weyl's law for the Neumann Laplacian 
	together with Proposition~\ref{prp:hdotneumann} in Appendix~\ref{prooffemconvergence} and
	the min-max principle. 
\end{remark}

Our goal is to obtain approximations of the covariance operator $L^{-2\beta}$ of the Gaussian random field $u$ which solves equation \eqref{general_main_equation_continuous}. Let
$$\varrho^\beta(x,y) = \sum_{j=1}^\infty \lambda_j^{-2\beta} e_j(x)e_j(y).$$
Then, one can readily check, by \citet[][Theorem 3.10]{kernel_mercer}, that the covariance operator $L^{-2\beta}$ is a kernel operator, with kernel $\varrho^\beta(\cdot,\cdot)$. 
That is, for any ${f\in L_2(\mathcal{D})}$,  we have $(L^{-2\beta}f)(x) = \int_{\mathcal{D}} \varrho^\beta(x,y) f(y)\text{d}y$ for a.e. $x \in \mathcal{D}$.
It is well-known that there exists a centered square-integrable Gaussian random field $u$ that solves 
\eqref{general_main_equation_continuous} if, and only if, its covariance operator, $L^{-2\beta}$, 
has finite trace \citep[][Theorem 3.2.5]{SPDEbrownucla}. Under Assumptions \ref{assumpDomain} and \ref{assumpOperator},
one can use Weyl's law (Remark \ref{rem:WeylsLaw})
to show that $L^{-2\beta}$ has finite trace if, and only if, $\beta > d/4$. Hence, if $\beta>d/4$, then $u$ is a centered square-integrable Gaussian random field with covariance function
$\varrho^\beta(x,y) = E[u(x)u(y)],$
where the equality holds for a.e. $(x,y)\in \mathcal{D}\times\mathcal{D}.$

\subsection{Finite element approximation}
\label{finite_element}

The goal is now to provide a convergence analysis for FEM approximations of the covariance operator $L^{-2\beta}$. Let us start by describing the setup we will use.

\begin{assumption}\label{assumpFE}
	Let $V_h \subset V$ be a finite element space that is spanned by a set of continuous piecewise linear basis functions $\{\varphi_j\}_{j=1}^{n_h}$ (see Appendix~\ref{sec:FEM_basis_func}), with $n_h\in\mathbb{N}$, defined with respect to a triangulation $\mathcal{T}_h$ of $\overline{\mathcal{D}}$ indexed by the mesh width $h:=\max_{T\in\mathcal{T}_h} h_T$, where ${h_T:= diam(T)}$ is the diameter of the element ${T\in \mathcal{T}_h}$. We  assume that the family $(\mathcal{T}_h)_{h\in(0,1)}$ of triangulations inducing the finite-dimensional subspaces $(V_h)_{h\in(0,1)}$ of $V$ is quasi-uniform, that is, there exist constants $K_1,K_2>0$ such that $\rho_T \geq K_1 h_T$ and $h_T\geq K_2 h$ for all $T\in \mathcal{T}_h$ and $h\in(0,1)$. Here, $\rho_T>0$ is the radius of the largest ball inscribed in $T\in\mathcal{T}_h$.
\end{assumption}

We are now in a position to describe the FEM discretization of the model \eqref{general_main_equation_continuous}. Let $L_h: V_h\to V_h$ be defined in terms of the bilinear form $a_L$ as its restriction to $V_h\times V_h$:
$$(L_h \phi_h, \psi_h)_{L_2(\mathcal{D})} = a_L(\phi_h, \psi_h) , \quad \phi_h, \psi_h \in V_h.$$ 
Note that $L_h$ is a positive-definite, symmetric, linear operator on the finite-dimensional space $V_h$. Hence, we may arrange the eigenvalues of $L_h$ as
${0 < \lambda_{1,h} \leq \lambda_{2,h} \leq \cdots\leq \lambda_{n_h, h},}$
with corresponding eigenvectors $\{e_{j,h}\}_{j=1}^{n_h}$ which are orthonormal in $L_2(\mathcal{D})$. Let $\mathcal{W}_h$ denote Gaussian white noise on $V_h$. That is, there exist independent standard Gaussian random variables $\xi_1,\ldots, \xi_{n_h}$ such that $\mathcal{W}_h = \sum_{j=1}^{n_h} \xi_j e_{j,h}$. Then, we refer to the following SPDE on $V_h$ as the discrete model of \eqref{general_main_equation_continuous}:

\begin{equation}\label{galerkinspde}
	L_h^{\beta} u_h = \mathcal{W}_h.
\end{equation}
Let $u_h$ be a solution of \eqref{galerkinspde}, then the covariance operator of $u_h$ is given by $L_h^{-2\beta}$, and
$$
\varrho_h^\beta(x,y) = \sum_{j=1}^{n_h} \lambda_{j,h}^{-2\beta} e_{j,h}(x) e_{j,h}(y),\quad\hbox{for a.e. $(x,y)\in\mathcal{D}\times \mathcal{D}$},
$$ 
is the corresponding covariance function. We have the following result regarding the convergence of the FEM approximation $ \varrho_h^\beta$ to the exact covariance function $\varrho^\beta$ in the $L_2(\mathcal{D}\times\mathcal{D})$-norm defined by $\|f\|_{L_2(\mathcal{D}\times\mathcal{D})}^2 = \int_{\mathcal{D}}\int_{\mathcal{D}} f(x,y)^2dxdy$. The proof is given in Appendix~\ref{prooffemconvergence}.

\begin{proposition}\label{cov_fem_approx_rate}
	Under Assumptions \ref{assumpDomain}, \ref{assumpOperator} and \ref{assumpFE}, for each $\beta>d/4$  and each $\varepsilon>0$, we have
	\begin{equation}\label{eq:cov_bound}
	\|\varrho^\beta - \varrho_h^\beta\|_{L_2(\mathcal{D}\times\mathcal{D})} \lesssim_{\varepsilon,\beta, \boldsymbol{H}, \kappa,\mathcal{D}} h^{\min\{4\beta-d/2 -\varepsilon,2\}}.
	\end{equation}
\end{proposition}

Here, and in the remainder of the paper, the notation 
$A \lesssim_{\theta_1,\ldots,\theta_k} B$, where
$k\in\mathbb{N}$, means that there exists a constant $C$ depending on
$\theta_1,\ldots,\theta_k$ ($\theta_i, i=1,\ldots,k$, can be a parameter, a function, a domain,
etc.) such that
$A\leq C B.$

\begin{remark}
	\citet[][Theorem 1]{cox_Kirchner} proved the bound \eqref{eq:cov_bound} in 
	the case of homogeneous Dirichlet boundary conditions. They did not 
	provide a bound for the case of homogeneous Neumann boundary conditions. 
	Proposition~\ref{cov_fem_approx_rate} arrives at the same bound for the
	Neumann case.
	For this, we additionally require that $\essinf_{x\in\mathcal{D}} \kappa(x) \geq \kappa_0>0$
	and that the domain $\mathcal{D}$ is a convex polytope in the Neumann case. 
	As far as we know, this is a 
	new result. The key step in the proof is to obtain an analogous result 
	to \citet[][Lemma~2]{cox_Kirchner}, which is given by Proposition~\ref{prp:hdotneumann} 
	in Appendix~\ref{prooffemconvergence}. 
\end{remark}

\section{Rational approximation}
\label{rationalapprox}
Having introduced the FEM approximation, we are now ready to define the 
complete approximation of the covariance operator of the generalized 
Whittle--Mat\'ern fields. The approximation is obtained by combining a 
rational approximation of the fractional power of the covariance operator 
with the FEM approximation. We begin by introducing the method and then 
provide a theoretical justification by showing an explicit rate of convergence of the approximate covariance function to the correct one in the $L_2(\mathcal{D}\times\mathcal{D})$-norm.

In \cite{rational_spde}, the authors obtained an approximation of the solution 
to \eqref{general_main_equation_continuous}, which also implicitly defines an 
approximation of the corresponding covariance operator. However, as we have 
previously mentioned, this results in an approximation that is not 
implementable in \texttt{R-INLA}. Also, for statistical applications there is usually no 
need to have an approximation of the solution itself, since only the 
corresponding distribution matters for inference. With this in mind, we 
propose to directly approximate the covariance operator $L^{-2\beta}$. To 
this end,
we first split $L_h^{-2\beta} = L_h^{-\{2\beta\}} L_h^{-\lfloor2\beta\rfloor}$, where $\{x\}=x-\lfloor x\rfloor$ is the fractional part of $x$. Then, we approximate $L_h^{-\{2\beta\}}$ with a rational approximation. This yields an approximation
\begin{equation} \label{cov_operator_appro2}
	L_h^{-2\beta} \approx	L_{h,m}^{-2\beta} := L_h^{-\lfloor 2\beta\rfloor} p(L_h^{-1})q(L_h^{-1})^{-1}.
\end{equation}
Here, $p(L_h^{-1}) = \sum_{i=0}^{m} a_i L_h^{m-i}$ and $q(L_h^{-1}) = \sum_{j=0}^m b_j L_h^{m-j}$ are polynomials obtained from a rational approximation of order $m$ of the real-valued function $f(x) = x^{\{2\beta\}}$. That is,
$$
x^{\{2\beta\}} \approx \frac{\sum_{i=0}^{m} a_i x^{i}}{\sum_{i=0}^m b_i x^{i}}.
$$
Specifically, to obtain $\{a_i\}_{i = 0}^{m}$ and  $\{b_i\}_{i = 0}^{m}$, we approximate the function $f(x) = x^{\{2\beta\}}$ on the interval $[\lambda_{n_h, h}^{-1},\lambda_{1,h}^{-1}]$, which covers the spectrum of $L^{-1}_h$. The coefficients are computed as the best rational approximation in the $L_\infty$-norm, which, for example, can be obtained via the second Remez algorithm \citep{remez1934determination} or by the recent, and more stable, BRASIL algorithm \citep{hofreither2021algorithm}. See Appendix~\ref{sec:algo_idea} for details about this algorithm and a justification for the choice of using the best rational approximation in $L_{\infty}$-norm. 

By defining the covariance function
$$
\varrho_{h,m}^\beta(x,y) = \sum_{j=1}^{n_h} \lambda_{j,h}^{-\lfloor 2\beta \rfloor}p(\lambda_{j,h}^{-1})q(\lambda_{j,h}^{-1})^{-1} e_{j,h}(x) e_{j,h}(y),\quad\hbox{for a.e. $(x,y)\in\mathcal{D}$},
$$
we have that $\varrho_{h,m}^\beta$ is the kernel of the covariance operator $L_{h,m}^{-2\beta}$. There are two sources of errors when we consider $\varrho_{h,m}^\beta$ as an approximation of the true covariance function $\varrho^\beta$ of the generalized Whittle--Mat\'ern field: the FEM approximation and the rational approximation. The following proposition, whose proof is given in Appendix~\ref{prooffemconvergence}, shows that we have control of these two sources of errors via the FEM mesh width $h$ and the order of the rational approximation $m$.

\begin{proposition}\label{ra_bound}
	Let $\beta>d/4$. Under Assumptions \ref{assumpDomain}, \ref{assumpOperator} and \ref{assumpFE}, for every $\varepsilon>0$ and for sufficiently small $h$, we have:
	\begin{equation} \label{total_error_bound}
		\|\varrho_{h,m}^\beta - \varrho^\beta\|_{L_2(\mathcal{D}\times\mathcal{D})} \lesssim_{\varepsilon,\beta, \boldsymbol{H}, \kappa,\mathcal{D}} h^{\min\{4\beta-d/2 -\varepsilon,2\}} + \mathbbm{1}_{2\beta \notin \mathbb{N}} h^{-d/2} e^{-2\pi\sqrt{\{2\beta\} m}}.
	\end{equation}
\end{proposition}

\begin{remark}
	We can calibrate the accuracy of the rational approximation with the finite element error by
	choosing $m\in\mathbb{N}$ such that $m = \lceil(\min\{4\beta-d/2 - \varepsilon,2\} + d/2)^2 \frac{(\log{h})^2}{4\pi^2\{2\beta\}}\rceil.$    
	This ensures that the rate of convergence in \eqref{total_error_bound} is $\min\{4\beta-d/2-\varepsilon,2\}$. 
	See Section~\ref{sec:numexp} for further details on the choice of $m$.
\end{remark}

\section{GMRF representation}
\label{GMRF}
The goal of this section is to obtain a sparse matrix representation of the precision operator of the  rational approximation from the previous section, so that the methods in Section~\ref{sec:overview} can be used for computationally efficient sampling and likelihood evaluation.

The solution $u_h$ in \eqref{galerkinspde} at spatial location $\boldsymbol{s}$ can be represented as $u_h(\boldsymbol{s})=\sum_{j=1}^{n_h} w_j \varphi_j(\boldsymbol{s})$, where $\{w_j\}_{j = 1}^{n_h}$ are stochastic weights and $\{\varphi_j\}_{j = 1}^{n_h}$ are the piecewise linear finite element basis functions. 
We will now show how to represent $\boldsymbol{w}=[w_1,...,w_{n_h}]^\top$ as a sum of independent GMRFs, each with a sparse precision matrix. 
The key step is to apply a partial fraction decomposition in \eqref{cov_operator_appro2}:
\begin{equation}\label{cov_operator_sum_form2}
	L_{h,m}^{-2\beta} =L_h^{-\lfloor 2\beta \rfloor} \left(\sum_{i=1}^{m}  r_i  
	(L_h-p_i I_{V_h})^{-1} +k I_{V_h} \right).
\end{equation}
Here, $\{r_i\}_{i=1}^{m}$, $\{p_i\}_{i=1}^{m}$ and $k$ are real numbers, 
and $I_{V_h}$ is the identity operator mapping the finite element space to 
itself. 
Let $\boldsymbol{C}$ be the mass matrix with elements
$\boldsymbol{C}_{i,j} = (\varphi_i,\varphi_j)_{L_2(\mathcal{D})}$, 
and let 
$\boldsymbol{L}$ be the matrix obtained by the bilinear form 
$a_L(\cdot,\cdot)$ induced by the differential operator $L$, which has 
elements ${\boldsymbol{L}_{i,j} = (\boldsymbol{H}\nabla \varphi_i, \nabla \varphi_j )_{L_2(\mathcal{D})} + (\kappa^2 \varphi_i, \varphi_j )_{L_2(\mathcal{D})}}$. 
Then, we can use \eqref{cov_operator_sum_form2} to obtain the covariance matrix of $\boldsymbol{w}$ as (see Appendix~\ref{covmatrix} for a derivation):
\begin{equation}
	\label{cov_matrix_rational2}
	\mathbf{\Sigma}^R_{\boldsymbol{u}} = (\boldsymbol{L}^{-1}
	\boldsymbol{C})^{\lfloor 2\beta \rfloor} \sum_{i=1}^{m}
	r_i(\boldsymbol{L}-p_i\boldsymbol{C})^{-1}+\boldsymbol{K}_{\lfloor 2\beta \rfloor},
\end{equation}
where $\boldsymbol{K}_0 = k\boldsymbol{C}$ and $\boldsymbol{K}_{n} = k(\boldsymbol{L}^{-1}\boldsymbol{C})^{n-1}  \boldsymbol{L}^{-1}$
when $n \geq 1,n\in\mathbb{N} $.
In the Mat\'ern case, that is, when $\kappa$ is a constant 
and $\boldsymbol{H}$ is an identity matrix, we simply have 
${\boldsymbol{L} =  \boldsymbol{G} + \kappa^2 \boldsymbol{C}}$, where
$\boldsymbol{G}$ is the stiffness matrix with elements $ \boldsymbol{G}_{i,j} = (\nabla\varphi_i,\nabla\varphi_j)_{L_2(\mathcal{D})}$.

Since we have the same degree for numerator and denominator in the rational approximation, we can use the BRASIL algorithm  \citep{hofreither2021algorithm} to compute the coefficients $\{a_i\}_{i = 0}^{m}$ and $\{b_i\}_{i = 0}^{m}$ in \eqref{cov_operator_appro2} and thus the coefficients $\{r_i\}_{i = 0}^{m}$, $\{p_i\}_{i = 0}^{m}$ and $k$ in \eqref{cov_operator_sum_form2}. Another option, commonly used in practice, is to use a ``near best'' rational approximation. One such option, which was used in \citet{rational_spde}, and which is also implemented in the \texttt{rSPDE} package, is the Clenshaw–Lord Chebyshev–Pad\'e algorithm \citep{baker1996pade}. See Appendix~\ref{sec:algo_idea} for details about this algorithm.
Also, observe that the interval $[\lambda_{n_h, h}^{-1},\lambda_{1,h}^{-1}]$ where one should compute the rational approximation may vary with the parameters $\kappa$ and $\boldsymbol{H}$, and that recomputing the coefficients $\{a_i\}_{i = 0}^{m}$ and  $\{b_i\}_{i = 0}^{m}$ for different values of these parameters is not practical for implementations. To avoid this, recall from Assumption \ref{assumpOperator} that  $\kappa_0^2$ is a lower bound for the eigenvalues of $L$ in the case of Neumann boundary conditions and that $\lambda_1 \leq \lambda_{1,h}$ (see Proposition~\ref{eigenval} in Appendix~\ref{prooffemconvergence}). We can, then, re-scale the operator $L_h$ as $L_h/\kappa_0^2$ so that we can replace the interval $[\lambda_{n_h, h}^{-1},\lambda_{1,h}^{-1}]$ by $[\delta, 1]$, where, ideally, $\delta$ is chosen in such way that $\delta \leq \kappa_0^2/\lambda_{n_h, h}$ for all considered mesh sizes $h$. In the \texttt{rSPDE} package, the choices $\delta = 0$ and  $\delta = 10^{-(5+m)/2}$ are implemented. However, the difference in accuracy with respect to approximating the covariance function is negligible between these two choices.  

For these  options, we verified empirically that if $f_\beta(x) = x^{\{2\beta\}}$, $\{2\beta\} = 2\beta - \lfloor 2\beta \rfloor$, and $\widehat{f}_{\beta,m}$ is the rational approximation of $f_\beta$ where the numerator and denominator have same degree $m$, then
$\widehat{f}_{\beta,m} = x^{\lfloor 2\beta \rfloor} \sum_{j=1}^m r_i (x - p_i)^{-1} + k,$
where $\{p_i\}_{i=1}^{m}$ are negative real numbers and $\{r_i\}_{i=1}^{m}$ and $k$ are positive real numbers. 
This, together with the fact that the BRASIL algorithm is only implemented for rational approximations with numerator and denominator having the same degree, are the main reasons we chose to consider the numerator and denominator having the same degree $m$. \citet{rational_spde} instead considered a rational approximation where the numerator has degree $m$ and the denominator has degree $m+1$. However, with this choice the partial fractions would not yield a decomposition into positive-definite operators in our case.

Since $\{p_i\}_{i=1}^{m}$ are negative real numbers and $\{r_i\}_{i=1}^{m}$ and $k$ are positive real numbers, we have that $r_i (\boldsymbol{L}^{-1}\boldsymbol{C})^{\lfloor 2\beta \rfloor}(\boldsymbol{L}-p_i\boldsymbol{C})^{-1}$, for $i=1,...m$, and $\boldsymbol{K}_{\lfloor 2\beta \rfloor}$ are valid covariance matrices. 
Thus, $\boldsymbol{w}$ can be expressed as a sum of $m+1$ independent random vectors $\boldsymbol{x}_i$ as in 
 \eqref{stoc_weights_sum_representaion}, where $\boldsymbol{Q}_i$ is the precision matrix of $\boldsymbol{x}_i$.
By \eqref{cov_matrix_rational2}, we obtain that
\begin{equation}\label{Q_matrix}
	\boldsymbol{Q}_i=\left \{
	\begin{array}{lcl}
		r_i^{-1}(\boldsymbol{L}-p_i \boldsymbol{C})(\boldsymbol{C}^{-1}\boldsymbol{L})^{\lfloor 2\beta \rfloor}     &      & {i = 1,...,m},\\
		\boldsymbol{K}^{-1}_{\lfloor 2\beta \rfloor}    &      & {i = m+1}.\\
	\end{array} \right. 
\end{equation}
Let $\boldsymbol{X}=[\boldsymbol{x}^\top_1,...,\boldsymbol{x}^\top_{m+1}]^\top$. Then, the precision matrix of $\boldsymbol{X}$ is the block diagonal matrix shown in \eqref{Q}.
The final step in order to obtain a GMRF representation is to use the mass lumping technique as for the standard SPDE approach, see Appendix C.5 in \citet{lindgren2011}. 
Thus, the mass matrix $\boldsymbol{C}$ in \eqref{Q_matrix} is replaced by a lumped mass matrix  $\tilde{\boldsymbol{C}}$, where $\tilde{\boldsymbol{C}}$ is a diagonal matrix with $ \tilde{\boldsymbol{C}}_{ii} = \sum_{j = 1}^{n_h}\boldsymbol{C}_{ij}$, for $i = 1,...,n_h$. 
With this adjustment, $\boldsymbol{Q}$ in \eqref{Q} is sparse and we thus have obtained a GMRF representation.

\section{Implementation and the \texttt{rSPDE} package}\label{package}

The proposed covariance-based rational approximation method has been implemented in the R package \texttt{rSPDE}. In the following sections, we will use this package to illustrate the performance of the method. In this section, we give a brief introduction to the package and how it can be used in combination with \texttt{R-INLA} for computationally efficient Bayesian inference of latent Gaussian models involving the generalized Whittle--Mat\'ern fields.

The usual workflow of fitting standard SPDE models in \texttt{R-INLA} can be divided into six steps. Namely, constructing the FEM mesh, defining SPDE model, creating a projector matrix, building the INLA stack, specifying the model formula, and finally calling the function \texttt{inla} to fit the model. Details about this can be found in \citet{lindgren2015software}.
To fit a model with a fractional SPDE, this procedure remains the same. The only difference is that when defining the SPDE model, creating the projector matrix and building the index for INLA stack, we use functions from the \texttt{rSPDE} package. These functions are very similar to the corresponding \texttt{R-INLA} functions in terms of functionality. For example, a fractional SPDE model can be created with the command: 
\begin{verbatim}
	model <- rspde.matern(mesh = mesh)
\end{verbatim}
where \texttt{mesh} is a FEM mesh that can be obtained by \texttt{inla.mesh.2d} function from \texttt{R-INLA}. The default order of the rational approximation in this function is $m=2$, which provides a good trade off between  computational cost and accuracy, see Figures~\ref{fig:L2normCov} and \ref{fig:L2normCov50}.
As for the corresponding \texttt{inla.spde2.matern} function that can be used to define non-fractional SPDE models in INLA, one can also set priors for $\kappa$ and $\tau$ in \texttt{rspde.matern}. Further, we can also define a prior for the smoothness parameter $\nu$ or specify $\nu$ so that a SPDE model with a fixed smoothness parameter can be generated. This feature can be used, for example, in the case that one already knows what $\nu$ is or wants to compare two different models with different $\nu$, as we will do in Section \ref{application}.

The projector matrix $\overline{\boldsymbol{A}}$ for a given mesh and observation locations \texttt{loc} is computed as
\begin{verbatim}
	A <- rspde.make.A(mesh = mesh, loc = loc)
\end{verbatim}
As for the creation of the model, the default order of the rational approximation when creating the projector matrix is $m = 2$, which can be changed by the user. The other arguments of the function are the same as those in the corresponding \texttt{R-INLA} function \texttt{inla.spde.make.A}.
In the step of building the INLA stack, usually an index set is needed. The index can be computed with the function \texttt{rspde.make.index}, which replaces the \texttt{R-INLA} function \texttt{inla.spde.make.index} and has the same arguments. With these functions, the fractional models can be used as any other random effect in \texttt{R-INLA}. 
After fitting the model with the \texttt{R-INLA} function \texttt{inla},  posterior samples from a latent field and hyperparameters
can be obtained by using \texttt{inla.posterior.sample}, and posterior distributions of the model parameters can be extracted via the \texttt{rSPDE} function \texttt{rspde.result}.

Besides the INLA-related functions, the \texttt{rSPDE} package also provides various utility functions. For example, once a fractional SPDE model, \texttt{model}, has been created with the \texttt{rspde.matern} function, one can simulate from it by calling \texttt{simulate(model)} to obtain a prior sample from a given choice of parameters, and the marginal log-likelihood from Section~\ref{sec:overview} can be computed by
\begin{verbatim}
	l <- rSPDE.matern.loglike(model, y, A, sigma.e)
\end{verbatim}
Here, \texttt{y} is the observed data and \texttt{sigma.e} is the standard deviation of the measurement noise. In addition, if a model is fitted with this approach, then kriging and posterior sampling can be obtained by using the \texttt{predict} function. For further details and examples, we refer the reader to the vignettes at \url{https://davidbolin.github.io/rSPDE}. 

Finally, \texttt{rSPDE} also provides an interface to the \texttt{inlabru} package \citep{inlabru}, which simplifies the construction of spatial models. This was used in the application in Section~\ref{application}, where
the entire code for defining and fitting the fractional model is:
\begin{verbatim}
	mesh <- inla.mesh.2d(loc = loc, max.edge = c(0.5, 10), cutoff = 0.35)
	spde <- rspde.matern(mesh = mesh, nu.upper.bound = 1)
	res <- bru(z ~ -1 + field(coordinates, model = spde), data = data)
\end{verbatim}
Here \texttt{loc} are the measurement locations and \texttt{data} is a data frame with the locations and observations.

\section{Numerical experiments}\label{sec:numexp}

In this section, we compare the accuracy of the covariance--based rational approximation with the operator-based method from \cite{rational_spde}, and with the ``parsimonious'' method from \cite{lindgren2011}. Since the latter method is implemented in R-INLA, we refer to it as the INLA approximation. We also note that the INLA method constructs a covariance-based Markov approximation \citep[see also][Section 2]{rational_spde}, so it can be viewed as a $0$th order covariance-based rational approximation. 

For the comparison, we consider the SPDE model \eqref{eq:spde_intro} with homogeneous zero Neumann boundary conditions on the unit square $\mathcal{D} = [0,1]^2$, with $\tau$ chosen such that $\sigma^2$ in the Mat\'ern covariance is one. 
The reason we consider the square domain, is that we have an explicit expression for the covariance function of the solution $u$. Indeed, we have, from \citet[Eq. (2.13)]{khristenkoetal}, that the covariance function of $u$ is given by
\begin{equation}
	\begin{split}\label{eq:foldedmatern}
		\varrho_u^\beta\left(\mathbf{x},\mathbf{y}\right) = \sum_{\mathbf{k} \in \mathbb{Z}^2} &\Big[\varrho(\|\mathbf{x} + 2\mathbf{k} - \mathbf{y}\|)
		+\varrho(\|(x_1+2k_1-y_1,x_2+2k_2+y_2)\|)\\
		&+\varrho(\|(x_1+2k_1+y_1,x_2+2k_2-y_2)\|)
		+\varrho(\|\mathbf{x} + 2\mathbf{k} + \mathbf{y}\|)\Big],
	\end{split}
\end{equation}
where $\|\cdot\|$ is the Euclidean norm on $\mathbb{R}^2$ and $\varrho(\cdot)$ is the Mat\'ern covariance function in \eqref{eq:matern} with $\sigma=1$ and $\nu = 2\beta - 1$.
To compare the accuracy of the covariance approximations, we evaluate the true and approximate covariance functions on a regular mesh on $[0,1]^2$ with $N =  100$ equally spaced nodes on each axis. 
We will compare these approximations with respect to the ${L_2([0,1]^2\times [0,1]^2)}$-norm and the supremum norm on $[0,1]^2\times [0,1]^2$. Appendix~\ref{sec:plot_likelihooderror} shows how we approximate the errors in these two norms in detail.

\begin{figure}[t]
	\centering 
	\includegraphics[width=1\textwidth]{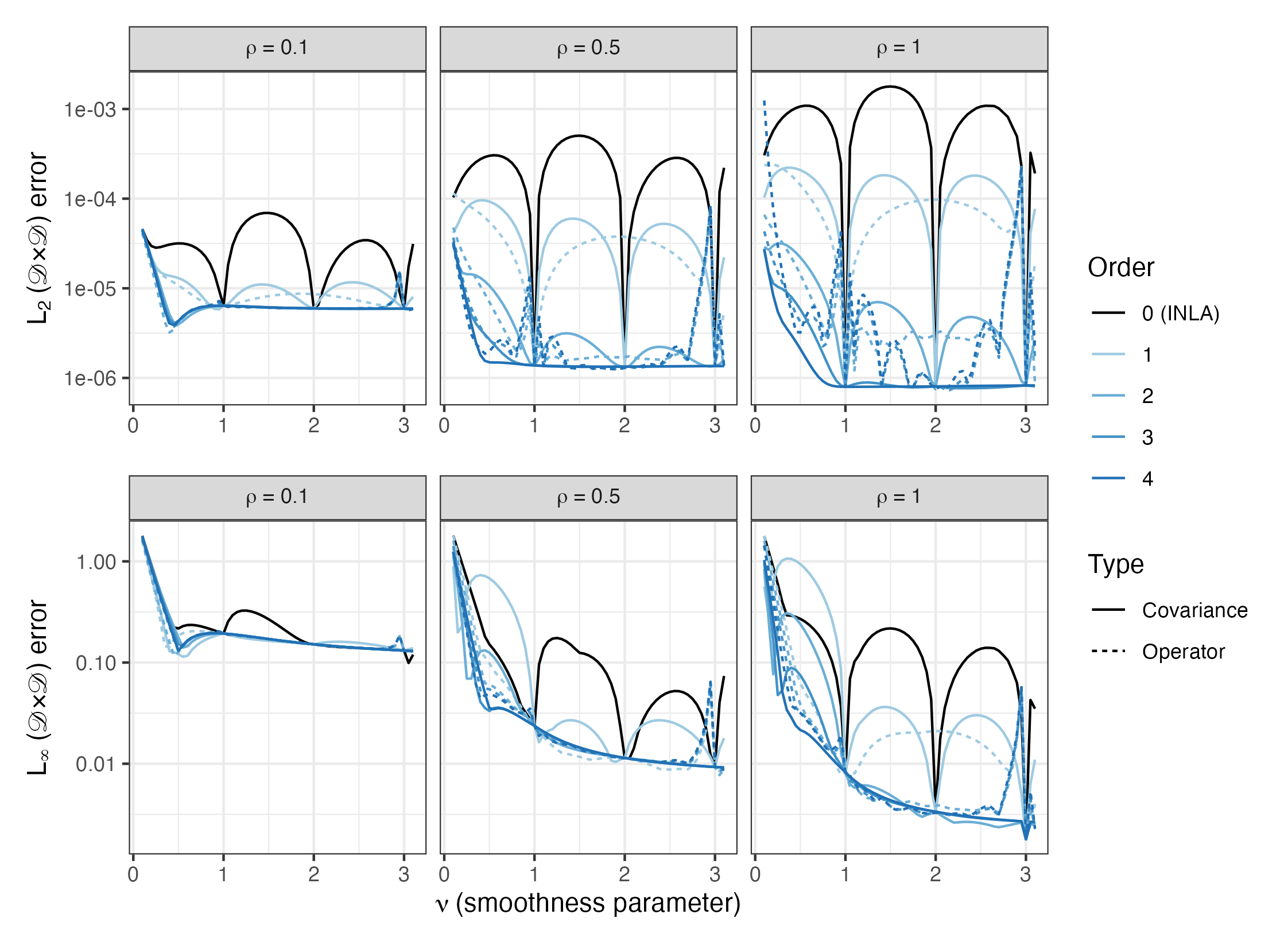}
	\caption{Errors in $L_2(\mathcal{D}\times \mathcal{D})$-norm (top) and supremum norm ($L_{\infty}(\mathcal{D}\times \mathcal{D})$) (bottom) on $\mathcal{D} = [0,1]^2$ for different practical ranges $\rho$ for different values of $\nu$. All methods use the same FEM mesh, with $100$ equally spaced nodes in each direction.}
	\label{fig:L2normCov}
\end{figure}

For the operator-based and covariance-based rational approximations, we consider the orders of rational approximation as $m = 1,2,3,4$. We choose smoothness parameters ranging from $0.1$ to $3.1$ with steps of size $0.05$. Further, we test three possible values of $\kappa$. These values of $\kappa$, say $\kappa_1(\nu), \kappa_2(\nu)$ and $\kappa_3(\nu)$ are chosen in such a way that the practical range $\rho = \sqrt{8\nu}/\kappa$ is fixed as $0.1, 0.5$ and $1$, respectively, for all values of $\nu$. The resulting errors for the different methods are shown in Figure \ref{fig:L2normCov}. 

We begin by observing that for smoothness parameters $\nu=1,2$ or $3$, there is no rational approximation and the errors only come from the FEM approximation. With this in mind, one should note that for smaller range parameters most of the approximation error comes from the FEM approximation, thus yielding a small difference of errors across the different methods. However, for larger ranges, such as, in this case, practical range equal to $1$, the errors have different orders of magnitude as the order of the rational approximation increases, with the errors from the operator-based and covariance-based approximations of same rational approximation order having approximately the same order of magnitude. Furthermore, we can observe numerical instabilities of the operator-based approximations of order 3 and 4 as $\nu$ increases for both practical ranges $0.5$ and $1$, whereas the covariance-based method is stable for all orders of approximation. 

In order to further illustrate the effect of the FEM error on the rational approximation of the covariance operator we repeated the analysis from above but with a coarser FEM mesh, consisting of 50 equally spaced nodes on each axis over the domain $[0,1]^2$. The results are shown in Figure \ref{fig:L2normCov50} in Appendix~\ref{sec:plot_likelihooderror}. We now observe that for practical range $0.1$, there is no visible difference between the covariance-based or operator-based rational approximations of orders 1 to 4, with a very small difference between the ``parsimonious'' INLA approximation and the remaining rational approximations. Further, for practical ranges $0.5$ and $1$, we hardly see any differences between the rational approximations of orders $2$, $3$ and $4$. The only noticeable difference being that for large values of $\nu$, the operator-based rational approximation becomes numerically unstable. On the other hand, it is noteworthy that for practical ranges $0.5$ and $1$, there is a significant difference (difference in orders of magnitude) between the rational approximations of order 0, 1 and the remaining orders.

To summarize, the results indicate that the covariance-based method generally has a similar accuracy as the operator-based method, which is higher than the accuracy provided by INLA's method. The results also show that the covariance-based method is more numerically stable, especially for larger values of $m$, the order of the rational approximation.

It is important to remember that the INLA method only provides a fixed approximation, furthermore it only works in this case of stationary parameters, whereas the other methods are applicable also for non-stationary models and can be made arbitrarily precise by increasing the order $m$. As previously mentioned, the operator-based method is not suitable for inference in R-INLA, but the covariance-based method is. Thus, in conclusion, the covariance-based method provides a method that facilitates inference for stationary and non-stationary fractional SPDE-based models in R-INLA, which is also more accurate than the current INLA method for stationary models. 

The numerical experiments in this section were implemented using the \texttt{rSPDE} package. All plots in this section, along with several more, for different choices of all the parameters involved, can be found in a \texttt{shiny} \citep{shiny} app available at \url{https://github.com/davidbolin/rSPDE}. The results above were obtained by the Clenshaw–Lord Chebyshev–Pad\'e algorithm with $\delta = 0$ (see Section \ref{GMRF}). 
The \texttt{shiny} app also contains the results by the BRASIL algorithm and the Clenshaw–Lord Chebyshev–Pad\'e algorithm with $\delta = 10^{-(5+m)/2}$.
We also include results on likelihood errors in Appendix~\ref{sec:plot_likelihooderror}.

\section{Application}
\label{application}

In this section, we illustrate the usage of the covariance-based rational approximation method through an application to a spatial data set of precipitation observations. 
The dataset, available at \url{https://www.image.ucar.edu/Data/precip_tapering/} contains annual precipitation anomalies observed by weather stations in the United States (standardized by the long-run mean and standard deviation for each station). We study the data from the year $1962$, which contains observations from 7352 stations throughout the contiguous United States.
We chose this dataset because it is simple enough to use a stationary model, which allows us to highlight the advantages of the fractional model without having to construct a complicated hierarchical model.
\citet{cov_taper} also studied this data as an illustration for the covariance tapering method.

\begin{figure}[t]
	\centering
	\includegraphics[width=0.8\textwidth]{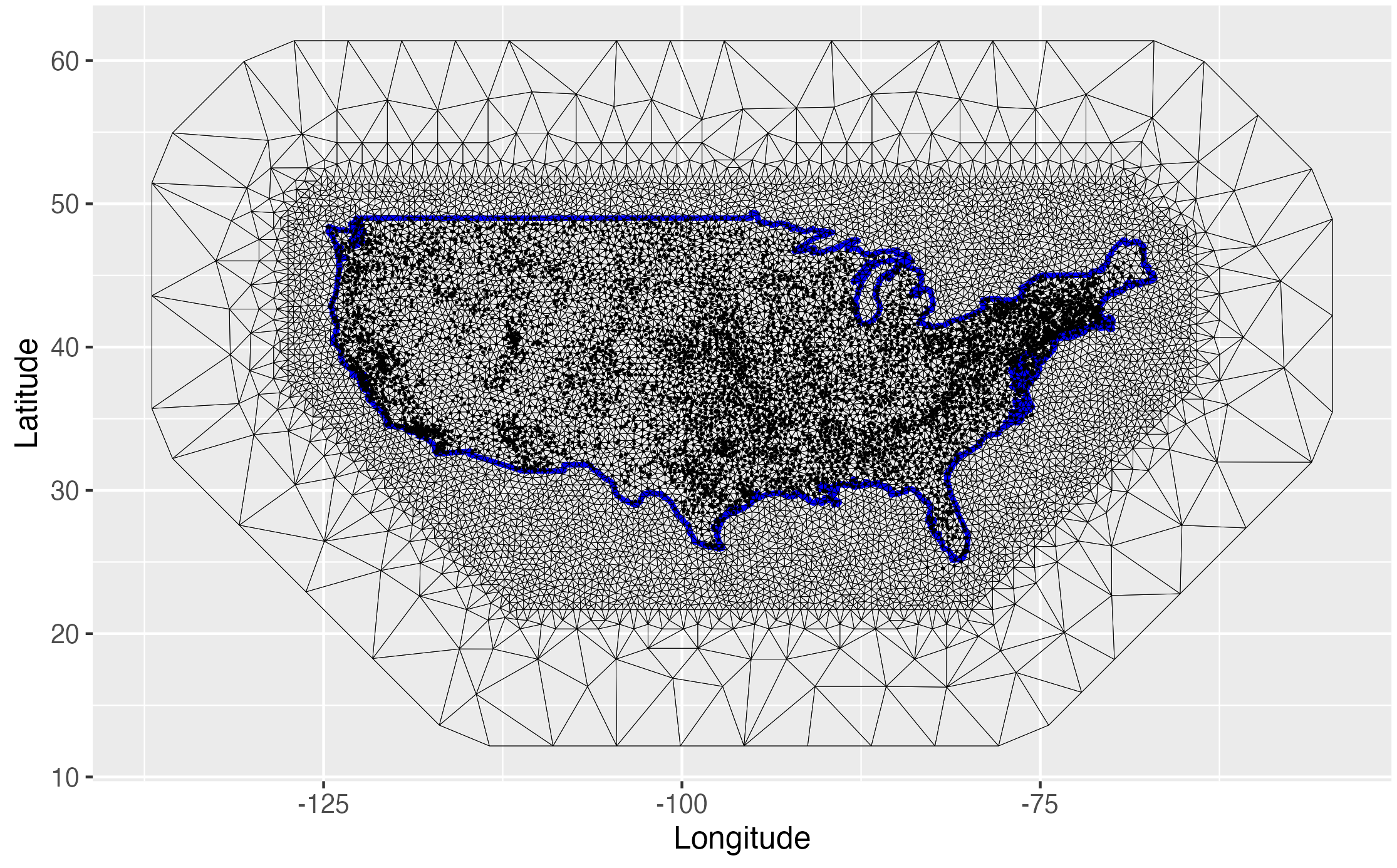}
	\caption{
		The finite element mesh over the contiguous US and the stations shown in dots.
	}
	\label{fig:map_mesh}
\end{figure}

We model the data by \eqref{hiermodel} where $\pmb{\epsilon}$ is independent Gaussian measurement noise with $\boldsymbol{Q}_{\pmb{\epsilon}} = \sigma_{\epsilon}^{-2}\mathbf{I}$ and $u$ is a Whittle-Mat\'ern field obtained as a solution to \eqref{eq:spde_intro}, where $\mathcal{D}$ is a bounded region (see Figure \ref{fig:map_mesh}).
The field is discretized using a finite element mesh that covers the contiguous United States with 9485 nodes.
Figure \ref{fig:map_mesh} shows the mesh and the 7352 stations. 
Our interest is to compare the stationary SPDE models with either a fractional smoothness parameter $\nu$ (referred to as the fractional model) or a fixed parameter $\nu = 1$ (referred to as the integer model) in terms of predictive power.
In order to more easily interpret the parameters, we consider a parameterization of the Whittle--Mat\'ern field in terms the standard deviation $\sigma =  \sqrt{\Gamma(\nu)} / (\tau \kappa^{\nu} \sqrt{(4\pi) \Gamma(\nu + 1)})$, the practical correlation range  $\rho = \sqrt{8\nu}/\kappa$, and the smoothness $\nu$.
The prior distributions for the parameters are chosen as the default choices from the \texttt{rSPDE} package.
That is, the priors of $\log(\rho)$ and $\log(\sigma)$ are independent Gaussian distributions with variance $10$ and the mean values are chosen based on size of the domain. Further, the prior of $\nu$ is a Beta distribution on the interval $(0, 1)$ with mean $1/2$ and variance $1/16$. The choice of prior for $\nu$ is motivated by the fact that we do not believe that this should be a very smooth field. We also tested with Beta distributions on larger intervals  and found that this did not affect the parameter estimates or the predictive performance of the model much.  

We fit the models using \texttt{R} (version 4.2.1) and the \texttt{rSPDE} package (version 2.2.0) combined with \texttt{inlabru} (version 2.7.0) and \texttt{R-INLA} (version 23.02.17) running on a machine with an Intel i9-12900KF CPU, 64GB RAM and an Ubuntu operating system.
The complete code for the analysis can be found in the supplementary materials. 
The total time for fitting the fractional and integer models are 38.4s and 15.4s, respectively.
\begin{figure}[t]
	\centering
	\includegraphics[width=0.8\linewidth]{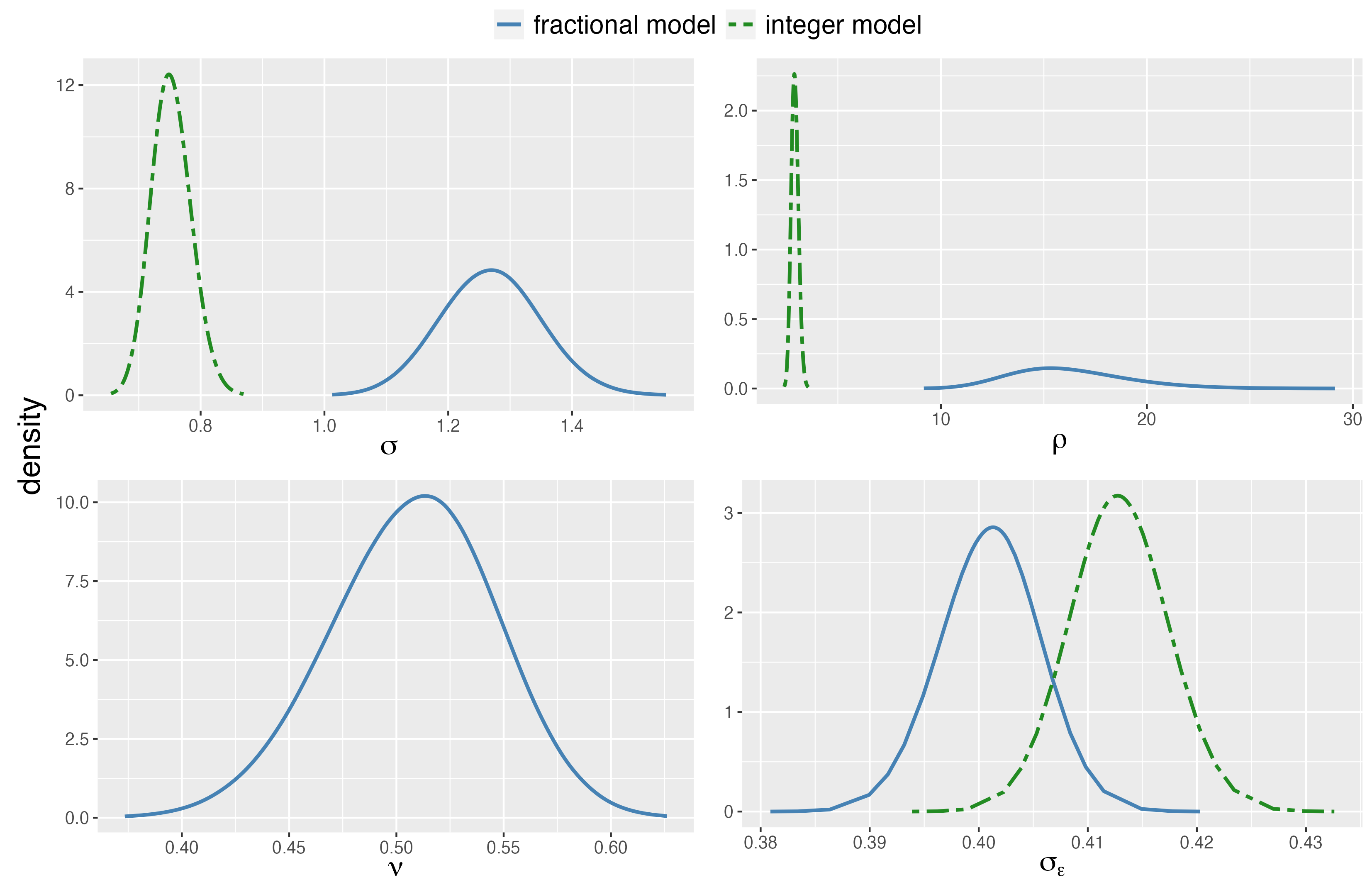}
	\caption{
		Posterior distributions of $\sigma$, $\rho$, $\nu$ and $\sigma_{\epsilon}$.
	}
	\label{posterior_plot}
\end{figure}

The posterior distributions of the parameters of the Gaussian field  and standard deviation of measurement noise for the two models are shown in Figure \ref{posterior_plot}. 
One can note that the posterior mode of $\nu$ for the fractional model is around $0.52$, which indicates that a fractional smoothness is needed. Compared to the fractional model, the integer model has a smaller $\sigma$ and a larger $\sigma_{\epsilon}$, indicating that the latent field explains less of the variability of the data. Finally, the practical correlation range  of the integer model is substantially smaller than that of the fractional model, which likely is caused by the fact that a small range is needed to better explain the short range behavior of the data if the smoothness parameter is forced to be an integer. 

To further compare the models, we perform two leave-group-out pseudo cross-validation studies \citep{lgocv}.
In the first, for each station, we predict the value of the station based on all data except that from stations that are closer than a certain distance $D$ (referred to as the distance of removed data).
We then vary this distance and compute the accuracy of the predictions as functions of $D$.
In the second, for each station, we instead remove the data from the $k$ nearest stations and compute the accuracy of the predictions as functions of $k$.
According to the screening effect \citep{screening_kriging}, in both cases the removed observations are the most informative.
The quality of the prediction is measured in terms of Mean Squared Errors (MSE) and the negative Log-Score (LS) \citep{logscore}.
Both metrics are negatively oriented, which means that a lower value indicates a better result.
The results of the two cross-validation studies are shown in Figure~\ref{CV_plot_distance}. 
We see that the fractional model outperforms the integer model in both cases. 
For example, the fractional model with the distance of removed data being 400km achieves the same levels of MSE and negative LS as the integer model with the distance of removed data being 300km (indicated by dashed lines).
Also, the fractional model with 125 removed data points achieves the same levels of MSE and negative LS as the integer model with 100 removed data points.
We can note that the two models have similar performance when the distance of removed data and number of removed data are close to zero. This is expected due to the mean-squared continuity of the latent fields, combined with the fact that the models have nugget effects. This means that both models will have an MSE close to the variance of measurement error when the distance of removed data or number of nearest removed data are close to zero.

\begin{figure}[t]
	\centering
		\includegraphics[width=0.8\linewidth]{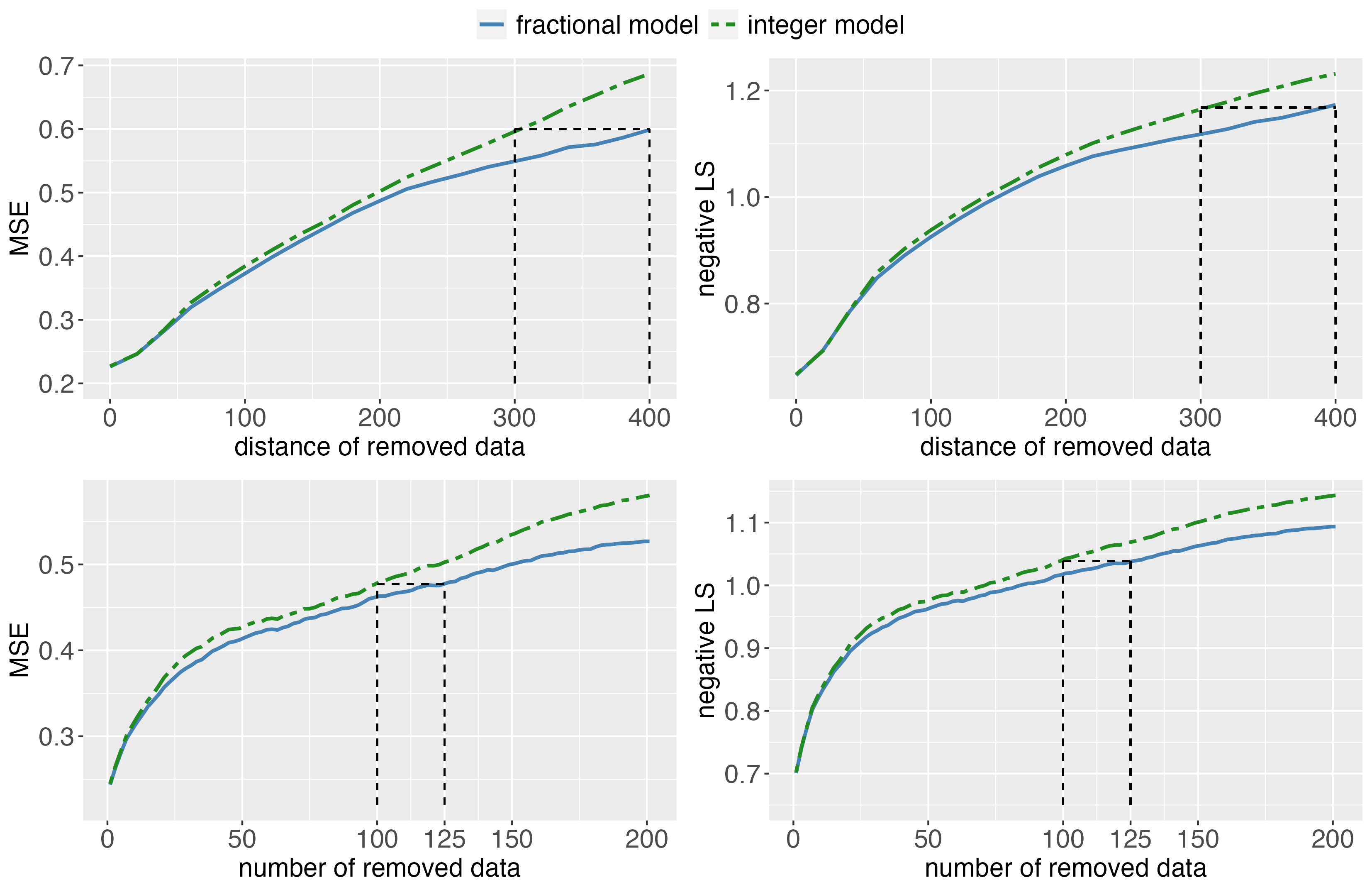}
	  \caption{MSE and negative Log-Score as functions of distance (in km) of removed data (top) and number of removed data (bottom). }
	\label{CV_plot_distance}
\end{figure}

\section{Discussion}
\label{discussion}

We have introduced a new rational SPDE approach which provides stable and computationally efficient approximations for the covariance structure of generalized Whittle--Mat\'ern Gaussian random fields with general smoothness $\beta > d/4$. We further derived an explicit rate of convergence of the method, which provides a theoretical justification for the approach. 
Compared to the rational SPDE approach of \cite{rational_spde}, the main advantage is that we obtain a GMRF representation of the approximation. This allowed us to implement the method so that fractional SPDE models now can be estimated in  \texttt{R-INLA}, where we in particular can estimate the smoothness parameter from data.  

The current version of \texttt{rSPDE} has truncated log-normal and beta priors for the smoothness parameter $\nu$ as possible choices. A natural question for future research is how this prior should be chosen in a more systematic way. A potential way to do this is following the idea of penalized complexity priors (PC-priors) \citep{simpson2017penalising}. \cite{fuglstad2019constructing} derived PC-priors for $\kappa$ and $\tau$ of the Whittle--Mat\'ern fields assuming a fixed value of $\nu$. We plan to extend that work by deriving PC-priors for all three parameters. 
Another potential area of future work is to extend the proposed method to spatio-temporal SPDE models as those proposed by \citet{bakka_diffusion-based_2020}.

\newpage
\begin{appendix}
\section{Additional notation}
\label{notation}

In this section, we introduce some notation that we will use for the technical details in the following sections.
Let $(E,\|\cdot\|_E)$ and $(F,\|\cdot\|_F)$ be two separable Hilbert 
spaces with norms $\|\cdot\|_E$ and $\|\cdot\|_F$ respectively. 
Then $(E,\|\cdot\|_E) \hookrightarrow (F,\|\cdot\|_F)$ means that 
$E \subset F$ and there exists a constant $C$ such that for any 
$x \in E$, we have $\|x\|_F \leq C\|x\|_E$. In this case, we say that 
$E$ is continuously embedded in $F$. If 
$(E,\|\cdot\|_E) \hookrightarrow (F,\|\cdot\|_F)\hookrightarrow 
(E,\|\cdot\|_E)$, we write 
$(E,\|\cdot\|_E) \cong (F,\|\cdot\|_F)$. We let $\mathcal{L}(E,F)$ denote the Banach space of bounded linear operators from $E$
to $F$ endowed with the operator norm,
that is, $\|A\|_{\mathcal{L}(E,F)} = \sup_{\|u\|_{E}=1} \|Au\|_{F}$,
where $A\in \mathcal{L}(E,F)$. Similarly, we let $\mathcal{L}_2(E,F)$ denote
the Banach space of Hilbert-Schmidt operators, endowed with the Hilbert-Schmidt
norm, that is, $\|A\|_{\mathcal{L}_2(E,F)}^2 = \sum_{i\in\mathbb{N}} \|Ae_i\|_{F}^2$,
where $\{e_i\}_{i\in\mathbb{N}}$ is a complete orthonormal set in 
$(E,\|\cdot\|_{E})$ and $A\in \mathcal{L}_2(E,F)$. 
We let $\mathcal{L}(E)$ denote $\mathcal{L}(E,E)$, with  norm $\|\cdot\|_{\mathcal{L}(E)}$,
and $\mathcal{L}_2(E)$ denote $\mathcal{L}_2(E,E)$,
with  norm $\|\cdot\|_{\mathcal{L}_2(E)}$. 
At last, if $E\subset F$, we let $I_{E,F}$ denote the inclusion map 
from $E$ to $F$.

Recall that a bounded linear operator $T$ on a Hilbert space $E$ is self-adjoint if, for all $f,g \in E$, $(T f, g )_{E} = (f,T g)_{E}$. Now let us show that $L^{-1}$ is self-adjoint, where $L$ is the operator from Section~\ref{sec:continuous_model}. We have that
$L^{-1}$ is compact and thus bounded on $L_2(\mathcal{D})$. 
For any $g \in L_2(\mathcal{D})$, $L^{-1}g\in V$ ($V$ is defined in Section~\ref{sec:continuous_model}). 
Thus, for any $f,g \in L_2(\mathcal{D})$, let $Lu = f$, so that $u = L^{-1}f$. By the symmetry of the bilinear form \eqref{bilinear_form}, we have that $L^{-1}$ is self-adjoint on $L_2(\mathcal{D})$:
$$(f,L^{-1}g)_{L_2(\mathcal{D})} = a_L(u,L^{-1}g)_{L_2(\mathcal{D})} = a_L(L^{-1}g,u)_{L_2(\mathcal{D})} = (LL^{-1}g,u)_{L_2(\mathcal{D})} = (g,L^{-1}f)_{L_2(\mathcal{D})}.$$

\section{Proofs of Proposition \ref{cov_fem_approx_rate} and  Proposition \ref{ra_bound}}\label{prooffemconvergence}
Let us start by providing some relations between the eigenvalues of $L$ and $L_h$. Recall, from Section \ref{sec:prelim}, that $\{\lambda_j\}_{j\in\mathbb{N}}$ are eigenvalues of $L$ and $\{\lambda_{j,h}\}_{j=1}^{n_h}$ are eigenvalues of $L_h$, both given in non-decreasing order. We have the following standard result:
\begin{proposition}\label{eigenval}
	Under Assumption \ref{assumpFE}, we have that 1. $\lambda_{n_h,h} \lesssim \lambda_{n_h} \lesssim n_h^{2/d}$ for sufficiently small $h\in (0,1)$ \citep[][Theorem 6.1]{strang_fix}; 2. $\lambda_j\leq \lambda_{j,h}$  (due to the min-max principle); and 3. $n_h \lesssim h^{-d}$ (due to quasi-uniformity of the triangulation).
\end{proposition}

Let, now, 
$\dot{H}_L^\sigma(\mathcal{D}) := \mathscr{D}(L^{\sigma/2}) = \bigl\{\psi\in L_2(\mathcal{D}):  \sum_{j\in\mathbb{N}}  \lambda_j^{\sigma} \<\psi, e_j\>_{L_2(\mathcal{D})}^2<\infty\bigr\},$
with inner product and norm given by
$$(\psi,\phi)_{\dot{H}_L^\sigma(\mathcal{D})} = (L^{\sigma/2}\psi, L^{\sigma/2}\phi)_{L_2(\mathcal{D})} = \sum_{j\in\mathbb{N}}{\lambda_j^\sigma} \<\psi, e_j\>_{L_2(\mathcal{D})}\<\phi, e_j\>_{L_2(\mathcal{D})}$$
and
$\|\psi\|_{\dot{H}_L^\sigma(\mathcal{D})}^2 = \<\psi, \psi\>_{\dot{H}_L^\sigma(\mathcal{D})}$, respectively.
Further, we define $[H_1,H_2]_\sigma$ as the real interpolation between the Hilbert spaces $H_1$ and $H_2$ \citep[see][Appendix A for a brief review of real interpolation of Hilbert spaces]{bolinsimaswallin}. 

We consider the fractional Sobolev space of order $\sigma$, with $0<\sigma<2$, $\sigma\neq 1$, given by 
$$H^\sigma(\mathcal{D}) = \begin{cases}
	[L_2(\mathcal{D}), H^1(\mathcal{D})]_{\sigma},&\hbox{for }0 < \sigma  <1,\\
	[H^1(\mathcal{D}),H^2(\mathcal{D})]_{\sigma-1},&\hbox{for } 1<\sigma < 2.
\end{cases}
$$

By \citet[][Lemma 2]{cox_Kirchner}, we have that with Dirichlet boundary conditions
$$(\dot{H}_L^\sigma, \|\cdot\|_{\dot{H}_L^\sigma(\mathcal{D})}) \cong ([L_2(\mathcal{D}), H^1_0(\mathcal{D})]_\sigma, \|\cdot\|_{[L_2(\mathcal{D}), H_0^1(\mathcal{D})]_\sigma}),\quad 0 < \sigma  < 1,$$
$$(\dot{H}_L^\sigma, \|\cdot\|_{\dot{H}_L^\sigma(\mathcal{D})}) \hookrightarrow (H^\sigma(\mathcal{D}), \|\cdot \|_{H^\sigma(\mathcal{D})}), 0 < \sigma <1,$$
where the norms $\|\cdot\|_{\dot{H}_L^\sigma(\mathcal{D})}$ and $\|\cdot\|_{H^\sigma(\mathcal{D})}$ are equivalent on $\dot{H}_L^\sigma(\mathcal{D})$ for $\sigma\neq 1/2$ and also
\begin{equation}\label{congruence}
	(\dot{H}_L^\sigma(\mathcal{D}),\|\cdot\|_{\dot{H}_L^\sigma(\mathcal{D})}) \cong (H^\sigma(\mathcal{D})\cap H_0^1(\mathcal{D}), \|\cdot\|_{H^\sigma(\mathcal{D})}),\quad 1\leq \sigma\leq 2.
\end{equation}

We want to apply \citet[][Theorem 1]{cox_Kirchner}, however it was only proved under Dirichlet boundary conditions. Therefore, we need some additional auxiliary results to conclude an analogous result in the case of Neumann boundary conditions. To this end, we need to prove the following result, which is a version of \citet[][Lemma~2]{cox_Kirchner} for Neumann boundary conditions:
\begin{proposition}\label{prp:hdotneumann}
	Under Neumann boundary conditions we have
	\begin{equation}\label{eq:hdotneuman1}
		(\dot{H}^\sigma_L, \|\cdot\|_{\dot{H}_L^\sigma(\mathcal{D})}) \cong (H^\sigma(\mathcal{D}), \|\cdot\|_{H^\sigma(\mathcal{D})}),\quad 0\leq \sigma\leq 1,
	\end{equation}
	\begin{equation}\label{eq:hdotneumann2}
		(\dot{H}_L^\sigma, \|\cdot\|_{\dot{H}_L^\sigma(\mathcal{D})}) \cong ([H^1(\mathcal{D}), H^2_\mathcal{N}(\mathcal{D})]_{\sigma-1}, \|\cdot\|_{[H^1(\mathcal{D}), H_\mathcal{N}^2(\mathcal{D})]_{\sigma-1}}),\quad 1 \leq \sigma  \leq 2,
	\end{equation}
	where $H^2_\mathcal{N}(\mathcal{D})$ was defined in Section \ref{sec:prelim}. Moreover,
	\begin{equation}\label{eq:hdotneumann3}
		(\dot{H}_L^\sigma, \|\cdot\|_{\dot{H}_L^\sigma(\mathcal{D})}) \hookrightarrow (H^\sigma(\mathcal{D}), \|\cdot \|_{H^\sigma(\mathcal{D})}), 1 < \sigma <2,
	\end{equation}
	where the norms $ \|\cdot\|_{\dot{H}_L^\sigma(\mathcal{D})}$ and $\|\cdot\|_{H^\sigma(\mathcal{D})}$ are equivalent on $\dot{H}_L^\sigma(\mathcal{D})$ for $\sigma\neq 3/2$.
\end{proposition}

\begin{proof}
	First, observe that $\dot{H}_L^0 = L_2(\mathcal{D})$. Also, since the bilinear form $a_L$ is continuous, coercive and symmetric, $a_L$ is an inner product on $H^1(\mathcal{D})$, whose corresponding norm is equivalent to $\|\cdot\|_{H^1(\mathcal{D})}$. Now, by definition of $\|\cdot\|_{\dot{H}_L^1(\mathcal{D})}$, we have that for every $\phi\in H^1(\mathcal{D})$, $a_L(\phi,\phi) = \|\phi\|_{\dot{H}_L^1(\mathcal{D})}^2$. This means that the norm induced by $a_L$ coincides with the norm $\|\cdot\|_{\dot{H}_L^1(\mathcal{D})}$. This shows the equivalence between $\|\cdot\|_{\dot{H}_L^1(\mathcal{D})}$ and $\|\cdot\|_{{H}^1(\mathcal{D})}$.
	
	Now, observe that from Lax-Milgram's lemma, for every $i\in\mathbb{N}$, the eigenvector $e_i$ of $L$ belongs to $H^1(\mathcal{D})$ and satisfies $(e_i,e_j)_{\dot{H}_L^1(\mathcal{D})} = a_L(e_i,e_j) = \lambda_i \delta_{i,j},$
	where $\delta_{i,j}$ is the Kronecker's delta. For any $\phi \in \dot{H}^1_L(\mathcal{D})$, we have that $\phi = \sum_{i\in\mathbb{N}} a_i e_i$, with $\sum_{i\in\mathbb{N}} a_i^2 \lambda_i <\infty$ and ${a_i = (\phi,e_i)_{L_2(\mathcal{D})}}$. Since  $\sum_{i\in\mathbb{N}} a_i^2 \lambda_i <\infty$, the series $\sum_{i\in\mathbb{N}} a_i e_i$ is absolutely convergent in $H^1(\mathcal{D})$, which implies that it also converges in $H^1(\mathcal{D})$ since $H^1(\mathcal{D})$ is a Hilbert (complete) space. On the other hand, the series converges to $\phi$ in $L_2(\mathcal{D})$, so the series must converge to $\phi$ in $H^1(\mathcal{D})$ as well because of the inclusion $H^1(\mathcal{D}) \subset L_2(\mathcal{D})$.  Thus, $\phi\in H^1(\mathcal{D})$. 
	
	Conversely, let $\psi\in H^1(\mathcal{D})$ and observe that $\{e_i/\sqrt{\lambda_i}\}_{i\in\mathbb{N}}$ is a complete orthonormal set in $(H^1(\mathcal{D}), \|\cdot\|_{\dot{H}^1_L(\mathcal{D})})$. Therefore, we have that
	$\psi = \sum_{i\in\mathbb{N}} b_i \frac{e_i}{\sqrt{\lambda_i}},$
	where the coefficients are ${b_i = (\psi, e_i/\sqrt{\lambda_i})_{\dot{H}^1_L(\mathcal{D})} = a_L(\psi, e_i/\sqrt{\lambda_i})}$. By Parseval's identity,
	$\|\psi\|_{\dot{H}^1_L(\mathcal{D})}^2 = \sum_{i\in\mathbb{N}} b_i^2.$
	By equivalence of $\|\cdot\|_{\dot{H}_L^1(\mathcal{D})}$ and $\|\cdot\|_{{H}^1(\mathcal{D})}$, we have that there exists $C>0$ such that $\|\psi\|_{\dot{H}^1_L(\mathcal{D})} \leq C \|\psi\|_{H^1(\mathcal{D})}$. Thus, since $\psi\in H^1(\mathcal{D})$, we have that $\|\psi\|_{H^1(\mathcal{D})} <\infty$, which in turn implies that $\sum_{i\in\mathbb{N}} b_i^2 <\infty$. On the other hand, $\{e_i\}_{i\in\mathbb{N}}$ is a complete orthonormal set in $L_2(\mathcal{D})$, so $\psi = \sum_{i\in\mathbb{N}} a_i e_i$, with $a_i = (\psi, e_i)_{L_2(\mathcal{D})}$. Therefore, $b_i = \sqrt{\lambda_i} a_i$, which yields,
	$\sum_{i\in\mathbb{N}} \lambda_i a_i^2 = \sum_{i\in\mathbb{N}} b_i^2 <\infty$, thus $\psi \in \dot{H}^1_L(\mathcal{D}).$
	Hence $(\dot{H}^1_L(\mathcal{D}), \|\cdot\|_{\dot{H}^1_L(\mathcal{D})})\cong (H^1(\mathcal{D}), \|\cdot\|_{H^1(\mathcal{D})})$. 
	
	We obtain \eqref{eq:hdotneuman1} by the same arguments as in the proof of \citet[Corollary 10]{bolinsimaswallin}.
	Similarly, to prove \eqref{eq:hdotneumann2}, it is enough to show that $(\dot{H}^2_L(\mathcal{D}), \|\cdot\|_{\dot{H}^2_L(\mathcal{D})})\cong (H_\mathcal{N}^2(\mathcal{D}), \|\cdot\|_{H^2(\mathcal{D})})$. 
	To this end, first, let $\phi \in \dot{H}^2_L(\mathcal{D})$ and 
	write $\phi = \sum_{i\in\mathbb{N}} a_i e_i$, with $a_i = 
	(\phi, e_i)_{L_2(\mathcal{D})}$. Let, $\phi_N= \sum_{i=1}^N a_i e_i$ 
	and by linearity of $L$, we have that $L \phi_N = \sum_{i=1}^N a_i 
	\lambda_i e_i$. Now, observe that $\sum_{i\in\mathbb{N}} \lambda_i^2 
	a_i^2 <\infty$ implies $L \phi_N$ converges to some $g\in L_2(\mathcal{D})$. 
	On the other hand, since $L:\dot{H}^2_L(\mathcal{D}) \to L_2(\mathcal{D})$ 
	is self-adjoint, it is a closed operator. Therefore, $L\phi = g$. We now 
	apply $H^2(\mathcal{D})$-regularity of $L$ (Remark \ref{rem:H2regul}) to 
	conclude that $\phi \in H_{\mathcal{N}}^2(\mathcal{D})$. Finally, it follows 
	from the closed graph theorem that $(\dot{H}_L^2(\mathcal{D}), 
	\|\cdot\|_{\dot{H}_L^2(\mathcal{D})}) \hookrightarrow 
	(H^2_\mathcal{N}(\mathcal{D}) , \|\cdot\|_{H^2(\mathcal{D})}).$ Indeed, 
	first observe that $\lambda_j\to \infty$ as $j\to\infty$. This yields 
	$(\dot{H}_L^2(\mathcal{D}), \|\cdot\|_{\dot{H}_L^2(\mathcal{D})}) 
	\hookrightarrow (L_2(\mathcal{D}) , \|\cdot\|_{L_2(\mathcal{D})}).$ 
	Now, let $\phi_N \to 0$ in 
	$\dot{H}^2_L(\mathcal{D})$, then $\phi_N\to 0$ in $L_2(\mathcal{D})$. 
	On the other hand if $I_{\dot{H}^2_L(\mathcal{D}), H^2_{\mathcal{N}}(\mathcal{D})}(\phi_N) \to \phi$, then 
	$\|\phi_N - \phi\|_{L_2(\mathcal{D})}\leq \|\phi_N-\phi\|_{H^2(\mathcal{D})} 
	\to 0$. So, $\phi = 0$, since $\phi_N \to 0$ in $L_2(\mathcal{D})$. 
	By the closed graph theorem $I_{\dot{H}^2_L(\mathcal{D}), H^2_{\mathcal{N}}(\mathcal{D})}$
	is a bounded operator.
	
	Conversely, let $\psi \in H^2_\mathcal{N}(\mathcal{D})$. By the Kirszbraun theorem \citep{kirszbraun}, $\boldsymbol{H}$ can be extended to a Lipschitz function on $\mathbb{R}^d$ with the same Lipschitz constant. Denote this extension by $\widetilde{\boldsymbol{H}}$. Now, let $R>0$ be such that $\overline{\mathcal{D}} \subset B(0,R)$, where $B(0,R)$ stands for the ball with center $0$ and radius $R$ in $\mathbb{R}^d$. Since $\widetilde{\boldsymbol{H}}$ is uniformly continuous, it is bounded in $B(0,2R)$. Let, also, $\varphi \in C^\infty_c(B(0,2R))$, such that $\varphi\equiv 1$ in $B(0,R)$. Then, by convexity of $B(0,2R)$, $\varphi$ is Lipschitz and bounded. This implies that $\varphi \widetilde{\boldsymbol{H}}$ is Lipschitz, since it is the product of bounded Lipschitz functions, $\varphi \widetilde{\boldsymbol{H}}\in C_c(\mathbb{R}^d)$, and the restriction of $\varphi\widetilde{\boldsymbol{H}}$ to $\mathcal{D}$ is $\boldsymbol{H}$. Therefore, by \citet[Theorem 1.4.1.1]{grisvard}, $\boldsymbol{H}\nabla \psi \in (H^1(\mathcal{D}))^d$. In particular,
	$L\phi \in L_2(\mathcal{D})$. Thus,
	$L\phi = \sum_{i\in\mathbb{N}} b_i e_i, \quad \sum_{i\in\mathbb{N}} b_i^2 <\infty,$
	where $b_i = (L\phi, e_i)_{L_2(\mathcal{D})}$. We then apply Gauss-Green formula \citep[Theorem 1.5.3.1]{grisvard} twice together with the fact that $\phi$ and $e_i$ satisfy Neumann boundary condition, to conclude that
	$$b_i = (L\phi,e_i)_{L_2(\mathcal{D})} = (\phi,Le_i)_{L_2(\mathcal{D})} = \lambda_i (\phi,e_i)_{L_2(\mathcal{D})}.$$
	Now, if we write $\phi = \sum_{i\in\mathbb{N}} a_i e_i$, we obtain that $b_i = \lambda_i a_i$. Therefore, we have that 
	$ \sum_{i\in\mathbb{N}} \lambda_i^2 a_i^2 =\sum_{i\in\mathbb{N}} b_i^2 < \infty.$
	Hence, $\phi \in \dot{H}_L^2(\mathcal{D})$. Now, we repeat the same argument from the previous inclusion, to obtain that $(H^2_\mathcal{N}(\mathcal{D}) , \|\cdot\|_{H^2(\mathcal{D})}) \hookrightarrow (\dot{H}_L^2(\mathcal{D}), \|\cdot\|_{\dot{H}_L^2(\mathcal{D})})$ from the closed graph theorem. This proves \eqref{eq:hdotneumann2}.
	
	Note that $H^2_\mathcal{N}(\mathcal{D}) \hookrightarrow H^2(\mathcal{D})$. So, by combining \eqref{eq:hdotneumann2} with a similar argument to the one in the proof of \cite[Corollary 10]{bolinsimaswallin}, we obtain \eqref{eq:hdotneumann3}.
	Finally, observe that since $\mathcal{D}$ is Lipschitz, we have, by \citet[Theorem 8.1]{grisvardinterpolation}, that $[L_2(\mathcal{D}), H^2_\mathcal{N}(\mathcal{D})]_{1/2} \cong H^1(\mathcal{D})$. This identification together with \cite[Theorem 2.2, item (vii)]{chandlerwildeetal} imply the further identification $[H^1(\mathcal{D}), H^2_\mathcal{N}(\mathcal{D})]_{\gamma} \cong [L_2(\mathcal{D}), H_\mathcal{N}^2(\mathcal{D})]_{\frac{1+\gamma}{2}}$, $0<\gamma<1$. The equivalence of the norms $\|\cdot\|_{\dot{H}_L^\sigma(\mathcal{D})}$ and $\|\cdot\|_{H^\sigma(\mathcal{D})}$ for $\sigma\neq 3/2$ now follows from \eqref{eq:hdotneumann2}, the identification $[H^1(\mathcal{D}), H^2_\mathcal{N}(\mathcal{D})]_{\gamma} \cong [L_2(\mathcal{D}), H_\mathcal{N}^2(\mathcal{D})]_{\frac{1+\gamma}{2}}$, $0<\gamma<1$, and another application of \citet[Theorem 8.1]{grisvardinterpolation}.
\end{proof}

We are now in a position to obtain a version of Theorem 1 of \cite{cox_Kirchner} (more precisely, of Remark 8 in \cite{cox_Kirchner}) that works for both Dirichlet and Neumann boundary conditions.

\begin{remark}\label{rem:interpolation}
	From Assumptions \ref{assumpDomain} and \ref{assumpFE}, there exists a linear operator $\mathcal{I}_h:H^2(\mathcal{D})\to V_h$ such that for every $1\leq \theta <2$, $\mathcal{I}_h:H^\theta(\mathcal{D})\to V_h$ is a continuous extension and there exists a constant $C$ which only depends on $\kappa, \boldsymbol{H}$ and $\mathcal{D}$ such that
	
	$$\|I_{H^\theta(\mathcal{D}),L_2(\mathcal{D})} - \mathcal{I}_h\|_{\mathcal{L}(H^\theta(\mathcal{D}), L_2(\mathcal{D}))} \lesssim_{\kappa, \boldsymbol{H}, \mathcal{D}} h^\theta,$$
	where $1\leq\theta\leq  2$ 
	Indeed, this follows by \citet[Theorem~3.2.1]{ciarletfinitelement} together with \citet[Theorem 3.5]{chandlerwildeetal}.
\end{remark}

\begin{lemma}\label{boundnorm}
	Under Assumption \ref{assumpDomain},\ref{assumpOperator} and \ref{assumpFE}, we have that for every $\tau>0$
	$$\|L^{-\tau} - L_h^{-\tau}\Pi_h\|_{\mathcal{L}(H^{\gamma}(\mathcal{D}), L_2(\mathcal{D}))} \lesssim_{\varepsilon,\tau,\gamma, \kappa, \boldsymbol{H}, \mathcal{D}} h^{\min\{2\tau + \gamma -\varepsilon, 2\}},$$
	where $\Pi_h:L_2(\mathcal{D}) \to V_h$ the $L_2(\mathcal{D})$-orthogonal projection onto $V_h$,
	$0 \leq \gamma \leq 2$, $\gamma\neq 1/2$ for the Dirichlet case, or $\gamma\neq 3/2$ for the Neumann case, $\varepsilon >0$ is arbitrary and $h>0$ is sufficiently small.
\end{lemma}
\begin{proof}
	For the Dirichlet case, Assumptions \ref{assumpDomain}, \ref{assumpOperator} and \ref{assumpFE} from Section \ref{sec:continuous_model} together with \eqref{congruence} and Remark \ref{rem:interpolation} imply the required assumptions for Theorem 1 of \cite{cox_Kirchner}. The case when $\gamma = 0$ follows  by choosing $\tau = \beta, \alpha=1$ and $\sigma=\delta=0$ whereas the case when $0<\gamma \leq 2$, $\gamma\neq 1/2$, follows from choosing $\tau = \beta, \alpha=1, \delta = \gamma$ and $\sigma=0.$
	
	For the Neumann case, Assumptions \ref{assumpDomain}, \ref{assumpOperator} and \ref{assumpFE}, together with Remark \ref{rem:interpolation} and Proposition \ref{prp:hdotneumann} allows us to use the same proof of \citet[Theorem 1]{cox_Kirchner} to obtain the desired result, where we take $\tau = \beta, \alpha=1$ and $\sigma=\delta=0$ when $\gamma=0$, or $\tau = \beta, \alpha=1, \delta = \gamma$ and $\sigma=0$, when $0<\gamma \leq 2$, $\gamma\neq 3/2$.
\end{proof}

We define
$$\varrho_h^\beta(x,y) = \sum_{j=1}^{n_h} \lambda_{j,h}^{-2\beta} e_{j,h}(x) e_{j,h}(y),\quad\hbox{for a.e. $(x,y)\in\mathcal{D}$}.$$
Then,
$\|L_h^{-2\beta}\Pi_h\|_{\mathcal{L}_2(L_2(\mathcal{D}))} = \|L_h^{-2\beta}\|_{\mathcal{L}_2(V_h)}  = \|\varrho_h^\beta\|_{L_2(\mathcal{D}\times\mathcal{D})}.$
\begin{remark}
	Note that $\varrho_h^\beta$ is the covariance function of the stochastic process obtained as the solution of \eqref{galerkinspde}.
\end{remark}
Now we are ready to give the proof of Proposition \ref{cov_fem_approx_rate}.
\begin{proof}[Proof of Proposition \ref{cov_fem_approx_rate}]
	Observe that $L^{-2\beta} - L_h^{-2\beta}\Pi_h$ is a kernel operator with kernel $\varrho^\beta - \varrho_h^\beta$. Thus,
	$\|\varrho^\beta - \varrho_h^\beta\|_{L_2(\mathcal{D}\times\mathcal{D})} = \|L^{-2\beta} - L_h^{-2\beta}\Pi_h\|_{\mathcal{L}_2(L_2(\mathcal{D}))}.$
	Therefore, it is enough to obtain a bound for $\|L^{-2\beta} - L_h^{-2\beta}\Pi_h\|_{\mathcal{L}_2(L_2(\mathcal{D}))}$. Fix any $\varepsilon>0$. Now, let $0<\delta < \min\{\beta -d/4, \varepsilon/4\}$. Then, we have that
	
	\begin{align}\label{proof1}
		\|L^{-2\beta} - L_h^{-2\beta}\Pi_h\|_{\mathcal{L}_2(L_2(\mathcal{D}))} \leq& \left\|\left(L^{-(2\beta -d/4 -\delta)} - L_h^{-(2\beta -d/4-\delta)}\Pi_h\right)L_h^{-(d/4+\delta)}\Pi_h\right\|_{\mathcal{L}_2(L_2(\mathcal{D}))} \nonumber\\
		&+ \left\|L^{-(2\beta  -d/4 -\delta)}\left(L_h^{-(d/4+\delta)}\Pi_h - L^{-(d/4+\delta)}\right)\right\|_{\mathcal{L}_2(L_2(\mathcal{D}))}.
	\end{align}

	We begin by handling the term first term in the right-hand side of \eqref{proof1}.
	Recall that if $H$ is a Hilbert space and $A,B:H\to H$ are linear operators, then
	\begin{equation}\label{HS-H}
		\|AB\|_{\mathcal{L}_2(H)}\leq \|A\|_{\mathcal{L}(H)}\|B\|_{\mathcal{L}_2(H)}.
	\end{equation}
	Now, let $\tau = 2\beta -d/4-\delta > 0$ and apply Lemma \ref{boundnorm} (where we take the $\varepsilon$ in its statement as $\varepsilon/2$ and $\gamma = 0$) together with equation \eqref{HS-H} to obtain
	\begin{align*}
		\left\|\left(L^{-\tau} - L_h^{-\tau}\Pi_h\right)L_h^{-(d/4+\delta)}\Pi_h\right\|_{\mathcal{L}_2(L_2(\mathcal{D}))} &\leq \|L^{-\tau} - L_h^{-\tau}\Pi_h\|_{\mathcal{L}(L_2(\mathcal{D}))} \|L_h^{-(d/4+\delta)}\Pi_h\|_{\mathcal{L}_2(L_2(\mathcal{D}))}  \\
		&\lesssim_{\varepsilon, \beta, \kappa, \boldsymbol{H}, \mathcal{D}} h^{\min\{4\beta -d/2 -2\delta -\varepsilon/2, 2\}} \|L_h^{-(d/4+\delta)}\Pi_h\|_{\mathcal{L}_2(L_2(\mathcal{D}))}.
	\end{align*}
	
	Let $\zeta(s) = \sum_{j=1}^\infty j^{-s}$ and $\theta = d/4+\delta$. We have the following bound for the Hilbert-Schmidt norm of $L^{-(d/4+\delta)}_h \Pi_h$:
	$$\|L_h^{-\theta} \Pi_h\|_{\mathcal{L}_2(L_2(\mathcal{D}))}^2 =  \sum_{j=1}^{n_h}  \lambda_{j,h}^{-2\theta} \leq \sum_{j=1}^{n_h} \lambda_j^{-2\theta} \lesssim_{\kappa, \boldsymbol{H}, \mathcal{D}} \sum_{j=1}^{n_h} j^{-4\theta/d} < \zeta(4\theta/d) <\infty,$$
	where we used item 2 of Proposition \ref{eigenval} and Weyl’s law (Remark \ref{rem:WeylsLaw}).
	Therefore, since $2\delta < \varepsilon/2$ and $h$ is sufficiently small, we obtain
	\begin{equation}\label{bound1}
		\left\|\left(L^{-(2\beta -d/4 -\delta)} - L_h^{-(2\beta -d/4-\delta)}\Pi_h\right)L_h^{-(d/4+\delta)}\Pi_h\right\|_{\mathcal{L}_2(L_2(\mathcal{D}))} \lesssim_{\varepsilon, \beta,\kappa, \boldsymbol{H}, \mathcal{D}} h^{\min\{4\beta-d/2-\varepsilon, 2\}}.
	\end{equation}
	
	Now let us give a bound for the second term on the right-hand side of 
	\eqref{proof1}. Let $\gamma = \min\{4\beta-d-4\delta, 2\}>0$, so 
	$\gamma\leq 2$. Observe that in order to apply Lemma \ref{boundnorm}, 
	we must choose $\delta$ such that  $\gamma \neq 1/2$ in the Dirichlet 
	case, or $\gamma\neq 3/2$ in the Neumann case. This is possible, since 
	we can reduce $\delta$ if necessary.
	
	The natural domain of the operator $L^{-(2\beta  -d/4 -\delta)}$ is 
	$ L_2(\mathcal{D})$. Furthermore, by the definition of the 
	$\dot{H}^\sigma_L(\mathcal{D})$ space, we have that $L^{-(2\beta  -d/4 
		-\delta)}: L_2(\mathcal{D}) \to \dot{H}_L^\gamma(\mathcal{D})$ since for 
	every $v \in L_2(\mathcal{D})$, $L^{-(2\beta  -d/4 -\delta)}v \in 
	\dot{H}_L^{4\beta-d/2-2\delta}(\mathcal{D})\subset \dot{H}_L^\gamma
	(\mathcal{D})$. If we restrict the domain of $L_h^{-\theta} \Pi_h$ to 
	$\dot{H}_L^\gamma(\mathcal{D})$, we also have that $L_h^{-\theta} \Pi_h: 
	\dot{H}_L^\gamma(\mathcal{D}) \to L_2(\mathcal{D})$. Let $A=L_h^{-\theta} 
	\Pi_h-L^{-\theta}$ and $B=L^{-(2\beta  -d/4 -\delta)}$. 
	Observe that
	\begin{equation*}\label{norm_ineq}
		\left\|BA\right\|_{\mathcal{L}_2(L_2(\mathcal{D}))} \leq 	\left\|B\right\|_{\mathcal{L}_2(L_2(\mathcal{D}),\dot{H}_L^\gamma(\mathcal{D}))}
		\left\|A\right\|_{\mathcal{L}(\dot{H}_L^\gamma(\mathcal{D}),L_2(\mathcal{D}))}.
	\end{equation*}
	
	Let us now show that $\left\|B\right\|_{\mathcal{L}_2(L_2(\mathcal{D}),\dot{H}_L^\gamma(\mathcal{D}))}$ is bounded. Recall that $\{e_j\}_{j\in\mathbb{N}}$ is an orthonormal basis in $L_2(\mathcal{D})$. Then we have
	\begin{align*}
		\left\|L^{-(2\beta  -d/4 -\delta)}\right\|_{\mathcal{L}_2(L^2(\mathcal{D}),\dot{H}_L^\gamma(\mathcal{D}))}^2
		&=
		\sum_{j=1}^\infty \left\|L^{-(2\beta  -d/4 -\delta)}e_j\right\|_{\dot{H}_L^\gamma(\mathcal{D})}^2
		=
		\sum_{j=1}^\infty \left\|L^{\gamma/2} L^{-(2\beta  -d/4 -\delta)}e_j\right\|_{L_2(\mathcal{D})}^2\\
		&=
		\sum_{j=1}^\infty \lambda_j^{\gamma-4\beta + d/2 + 2\delta}
		\lesssim_{\kappa, \boldsymbol{H},\mathcal{D}}
		\sum_{j=1}^\infty j^{2\gamma/d-8\beta/d+4\delta/d+1},
	\end{align*}
	which converges since $2\gamma/d-8\beta/d+4\delta/d+1< -1$, and where the last inequality comes from $\gamma-2\delta<\gamma\leq 4\beta-d-4\delta$.
	
	Now let us handle the term $ \left\|A\right\|_{\mathcal{L}(\dot{H}_L^\gamma(\mathcal{D}),L_2(\mathcal{D}))}$. By (16) and (17) from Lemma 2 in \cite{cox_Kirchner} for the Dirichlet case, or by Proposition \ref{prp:hdotneumann} for the Neumann case, we can conclude that $\dot{H}_L^\gamma(\mathcal{D}) \subset H^\gamma(\mathcal{D})$ and $\left\|\cdot\right\|_{\dot{H}_L^\gamma(\mathcal{D})}$ is equivalent to $\left\|\cdot\right\|_{H^\gamma(\mathcal{D})}$ when $0 \leq \gamma \leq 2$ and $\gamma \neq 1/2$ for the Dirichlet case or $\gamma\neq 3/2$ for the Neumann case. By equivalency of the two norms, there exists a constant $C$ such that $\left\|v\right\|_{H^\gamma(\mathcal{D})} \leq C\cdot \left\|v\right\|_{\dot{H}_L^\gamma(\mathcal{D})}$, which implies $1/\left\|v\right\|_{\dot{H}_L^\gamma(\mathcal{D})} \leq C/\left\|v\right\|_{H^\gamma(\mathcal{D})}$, for every $v \in \dot{H}_L^\gamma(\mathcal{D})$. Then by $\dot{H}_L^\gamma(\mathcal{D}) \subset H^\gamma(\mathcal{D})$, we can conclude that $ \left\|A\right\|_{\mathcal{L}(\dot{H}_L^\gamma(\mathcal{D}),L_2(\mathcal{D}))} \leq C \cdot \left\|A\right\|_{\mathcal{L}(H^\gamma(\mathcal{D}),L_2(\mathcal{D}))}$. Combining this with Lemma~\ref{boundnorm}, we obtain that 
	\begin{align*}
		\left\| L_h^{-\theta} \Pi_h-L^{-\theta} \right\|_{\mathcal{L}(\dot{H}_L^\gamma(\mathcal{D}),L_2(\mathcal{D}))}
		&\lesssim_{\varepsilon,\theta,\gamma, \kappa, \boldsymbol{H}, \mathcal{D}}
		h^{\min\{2\theta + \gamma - \varepsilon/2,2\}}
		=
		h^{\min\{4\beta-d/2-2\delta-\varepsilon/2,2\}},
	\end{align*}
	where we chose $\varepsilon$ in the statement of Lemma  \ref{boundnorm} as $\varepsilon/2$. Again, since $2\delta < \varepsilon/2$ and $h$ is sufficiently small, we arrive at
	\begin{equation}\label{bound2}
		\left\| L_h^{-\theta} \Pi_h-L^{-\theta} \right\|_{\mathcal{L}(\dot{H}_L^\gamma(\mathcal{D}),L_2(\mathcal{D}))}
		\lesssim_{\varepsilon, \beta,\kappa, \boldsymbol{H}, \mathcal{D}} h^{\min\{4\beta-d/2-\varepsilon, 2\}}.
	\end{equation}
	
	The result now follows from \eqref{bound1} and \eqref{bound2}.
\end{proof}

\begin{proof}[Proof of Proposition \ref{ra_bound}]
	First, note that 
	$\|\varrho_{h,m}^\beta - \varrho^\beta\|_{L_2(\mathcal{D}\times\mathcal{D})} = \|L^{-2\beta} - L_{h,m}^{-2\beta}\Pi_h\|_{\mathcal{L}_2(L_2(\mathcal{D}))}$, and we similarly have that 
	$\|\varrho_{h,m}^\beta - \varrho_h^\beta\|_{L_2(\mathcal{D}\times\mathcal{D})} = \|L^{-2\beta}_{h,m}\Pi_h - L_{h}^{-2\beta}\Pi_h\|_{\mathcal{L}_2(L_2(\mathcal{D}))}$ 
	and also 
	$\|\varrho^\beta - \varrho_h^\beta\|_{L_2(\mathcal{D}\times\mathcal{D})} = \|L^{-2\beta} - L_h^{-2\beta}\Pi_h\|_{\mathcal{L}_2(L_2(\mathcal{D}))}$. 
	Therefore, by the triangle inequality, 
	$$\|\varrho_{h,m}^\beta - \varrho^\beta\|_{L_2(\mathcal{D}\times\mathcal{D})} \leq \|L_h^{-2\beta}\Pi_h - L_{h,m}^{-2\beta}\Pi_h\|_{\mathcal{L}_2(L_2(\mathcal{D}))} + \|L^{-2\beta} - L_h^{-2\beta}\Pi_h\|_{\mathcal{L}_2(L_2(\mathcal{D}))}.$$
	
	We begin by obtaining an upper bound for $\|L_h^{-2\beta}\Pi_h - L_{h,m}^{-2\beta}\Pi_h\|_{\mathcal{L}_2(L_2(\mathcal{D}))}$. Recall from Section \ref{finite_element}, that the eigenvalues of $L_h$ are $0<\lambda_{1,h} \leq \lambda_{2,h} \leq \cdots\leq \lambda_{n_h, h}$, 
	with corresponding eigenvectors $\{e_{j,h}\}_{j=1}^{n_h}$, which are orthonormal in $L_2(\mathcal{D})$. By item 2 of Proposition \ref{eigenval}, we have that $J_h \subset J$, where $J_h = [\lambda_{n_h,h}^{-1}, \lambda_{1,h}^{-1}]$ and $J = [0,\lambda_1^{-1}]$, since $\lambda_1$ is the smallest eigenvalue of $L$. We normalize $L$ so that $\lambda_1 \geq 1 $. Thus, $J_h \subset J \subset [0,1].$ Now, let $f(x) = x^{2\beta}$ and $\hat{f}(x) = x^{\{2\beta\}}$, where $\{2\beta\} = 2\beta-\lfloor 2\beta \rfloor$, so that $f(x) = x^{\lfloor 2\beta\rfloor}\hat{f}(x)$. Let $\hat{r}_h(x) = \frac{p(x)}{q(x)}$ be the $L_{\infty}$-best approximation of $\hat{f}(x)$ on $J_h$, and define $r_h(x) = x^{\lfloor 2\beta\rfloor}\hat{r}_h(x)$. Then, we have the following bound:
	\begin{eqnarray}
		\|L_h^{-2\beta}\Pi_h - L_{h,m}^{-2\beta}\Pi_h\|_{\mathcal{L}_2(L_2(\mathcal{D}))}^2 
		&=&
		\sum_{j = 1}^{n_h} \|L_h^{-2\beta} e_{j,h} -L_{h,m}^{-2\beta} e_{j,h} \|_{L_2(\mathcal{D})}^2 
		=
		\sum_{j = 1}^{n_h} (\lambda_{j,h}^{-2\beta}-r_h(\lambda_{j,h}^{-1}))^2
		\nonumber \\
		&\leq&
		\label{bound33}
		n_h \max\limits_{1 \leq j \leq n_h} \lvert \lambda_{j,h}^{-2\beta}-r_h(\lambda_{j,h}^{-1}) \rvert^2.
	\end{eqnarray}
	We now apply  \cite[Theorem 1]{Stahl}, and observe that  $x^{\lfloor 2\beta \rfloor} \leq 1$ on $J_h$, to obtain:
	\begin{equation}
		\max\limits_{1 \leq j \leq n_h} \lvert \lambda_{j,h}^{-2\beta}-r(\lambda_{j,h}^{-1}) \rvert
		\leq 
		\sup\limits_{x \in J_h} \lvert f(x)-r(x)\rvert 
		\leq 
		\sup\limits_{x \in [0,1]} \lvert \hat{f}(x)-\hat{r}(x)\rvert
		\lesssim 
		\label{bound36}
		e^{-2\pi\sqrt{\{2\beta\}  m}}.
	\end{equation}
	Thus, by \eqref{bound33} and \eqref{bound36}, we have
	$\|L_h^{-2\beta}\Pi_h - L_{h,m}^{-2\beta}\Pi_h\|_{\mathcal{L}_2(L_2(\mathcal{D}))} \lesssim n_h^{1/2}e^{-2\pi\sqrt{\{2\beta\}  m}}$ and by item 3 of Proposition \ref{eigenval}, we obtain
	$n_h^{1/2}e^{-2\pi\sqrt{\{2\beta\}  m}} \lesssim h^{-d/2}e^{-2\pi\sqrt{\{2\beta\}  m}}$. This source of error only occurs if we need the rational approximation, i.e., if $2\beta \notin \mathbb{N}$. Thus, combining this with the bound $\|L^{-2\beta} - L_h^{-2\beta}\Pi_h\|_{\mathcal{L}_2(L_2(\mathcal{D}))} \lesssim_{\varepsilon,\beta, \boldsymbol{H}, \kappa,\mathcal{D}} h^{\min\{4\beta-d/2 -\varepsilon,2\}}$ from Proposition~\ref{cov_fem_approx_rate}, yields:
	\begin{equation*}
		\|\varrho_{h,m}^\beta - \varrho^\beta\|_{\mathcal{L}_2(\mathcal{D}\times\mathcal{D})} 
		\lesssim_{\varepsilon,\beta, \boldsymbol{H}, \kappa,\mathcal{D}}
		\mathbbm{1}_{2\beta \notin \mathbb{N}}h^{-d/2}e^{-2\pi\sqrt{\{2\beta\}  m}} + h^{\min\{4\beta-d/2 -\varepsilon,2\}}.
	\end{equation*}
\end{proof}

\section{Derivation of the GMRF representation}\label{covmatrix}
In this section, we derive equation \eqref{cov_matrix_rational2}. Recall 
the rational approximated covariance operator in 
\eqref{cov_operator_sum_form2}: $L_{h,m}^{-2\beta} =
L_h^{-\lfloor 2\beta \rfloor} (\sum_{i=1}^{m}  r_i  
(L_h-p_i I_{V_h})^{-1} +k I_{V_h})$, where $L_h$ was defined in 
Section \ref{finite_element}, $L_{h,m}^{-2\beta}$ was defined in 
Section \ref{rationalapprox} and $I_{V_h}$ is the identity map on
the finite element space $V_h$. The first part of this expression is the 
sum of the terms of the form $r_i L_h^{-\lfloor 2\beta \rfloor} 
(L_h-p_i I_{V_h})^{-1}, i = 1,...,m$, whereas the second part is 
$k  L_h^{-\lfloor 2\beta \rfloor}$. Since $\{r_i\}_{i = 1}^m$ and $k$ 
are positive and $\{p_i\}_{i = 1}^m$ are negative real numbers, 
$\{r_i L_h^{-\lfloor 2\beta \rfloor} (L_h-p_i I_{V_h})^{-1}\}_{i=1}^m$ 
and $k  L_h^{-\lfloor 2\beta \rfloor}$ are positive-definite. 
They are also self-adjoint, and thus valid 
covariance operators.

We will deal with each term in the partial fractions expansion separately. 
We begin with the terms of the form $r L_h^{-\lfloor 2\beta\rfloor} 
(L_h-p I_{V_h})^{-1}$. Observe that this term is the covariance operator of 
the solution of the SPDE
$r^{-1/2}(L_h-p I_{V_h})^{1/2} L_h^{\lfloor 2\beta\rfloor/2} x = 
\mathcal{W}_h.$
If $\lfloor 2\beta\rfloor$ is odd, $\lfloor 2\beta\rfloor = 2n+1$, 
with $n\in\mathbb{N}$, we can rewrite the equation as
$r^{-1/2}((L_h-p I)L_h)^{1/2}L_h^{n} x = \mathcal{W}_h$, or equivalently
\begin{align}
	\label{twostep1}
	r^{-1/2}\hat{L}^{1/2} z = \mathcal{W}_h, \\
	\label{twostep2}
	L_h^n x= z,
\end{align}
where $\hat{L} = (L_h-p I_{V_h})L_h$ and $z \in V_h$ (see Section 
\ref{finite_element} for the definition of $V_h$). 

Let $\{\varphi_j\}_{j = 1}^{n_h}$ be the finite element basis of $V_h$. 
We can write $z$ in the finite element basis as $z  = \sum_{j = 1}^{n_h} 
z_j \varphi_j$. Similarly, we have that $x=\sum_{j=1}^{n_h} x_j \varphi_j$. 
Let us now obtain a relation between $\boldsymbol{z}=[z_1,...,z_{n_h}]^\top$ 
and $\boldsymbol{x}=[x_1,...,x_{n_h}]^\top$. Observe that, for each 
$l = 1,...,n_h$, we have $(z, \varphi_l)_{L_2(\mathcal{D})} = 
\sum_{j = 1}^{n_h} z_j(\varphi_j, \varphi_l)_{L_2(\mathcal{D})}$. However, 
by \eqref{twostep1} and \eqref{twostep2}, we also have $(z, 
\varphi_l)_{L_2(\mathcal{D})} = ( L_h^n x, \varphi_l )_{L_2(\mathcal{D})} 
=  \sum_{j = 1}^{n_h} x_j (  L_h^n \varphi_j, \varphi_l )_{L_2(\mathcal{D})}$. 
Let us now compute $(  L_h^n \varphi_j, \varphi_l )_{L_2(\mathcal{D})}$. 
To this end, let $\boldsymbol{B}$ be the matrix of the operator $L_h$ 
in the basis $\{\varphi_i\}_{i = 1}^{n_h}$ so that 
$ \varphi_j = \sum_{k = 1}^{n_h}\boldsymbol{B}_{j,k}\varphi_k$.  
Thus $(  L_h \varphi_j, \varphi_l ) _{L_2(\mathcal{D})} = \sum_{k = 1}^{n_h} 
\boldsymbol{B}_{j,k}(\varphi_k,\varphi_l)_{L_2(\mathcal{D})}$. Let, also, 
$\boldsymbol{L}_{j,l}: = a_L(\varphi_j, \varphi_l) = (L_h\varphi_j, 
\varphi_l)$ (recall the bilinear form $a_L(\cdot,\cdot)$ from Section 
\ref{sec:continuous_model}) and $\boldsymbol{C}_{j,l} = (\varphi_j,
\varphi_l)_{L_2(\mathcal{D})}$ (Both $\boldsymbol{L}$ and $\boldsymbol{C}$ 
are symmetric). Then, $\boldsymbol{B} = \boldsymbol{L}\boldsymbol{C}^{-1}$, 
and
\begin{align*}
	( L_h^n \varphi_j, \varphi_l ) _{L_2(\mathcal{D})} 
	&= 
	(  L_h^{n-1}(L_h \varphi_j), \varphi_l ) _{L_2(\mathcal{D})}
	=
	(  L_h^{n-1}\sum_{k = 1}^{n_h} \boldsymbol{B}_{j,k}\varphi_k, \varphi_l ) _{L_2(\mathcal{D})} \\
	&=
	\sum_{k = 1}^{n_h} \boldsymbol{B}_{j,k} (L_h^{n-1}\varphi_k,\varphi_l)_{L_2(\mathcal{D})}.
\end{align*}
The relation $\boldsymbol{z} = (\boldsymbol{C}^{-1}\boldsymbol{L})^n \boldsymbol{x}$ now follows by induction (the base case is $\boldsymbol{L}_{j,l} = (L_h\varphi_j, \varphi_l)$) since $(  L_h^n \varphi_j, \varphi_l ) _{L_2(\mathcal{D})} = [\boldsymbol{B}^{n-1}\boldsymbol{L}]_{j,l} = [(\boldsymbol{L}\boldsymbol{C}^{-1})^{n-1}\boldsymbol{L}]_{j,l}$. 

We are now ready to obtain the distribution of $\boldsymbol{x}$. Note that 
$\hat{L}^{1/2}: V_h\to V_h$ is an isomorphism: By the coerciveness of bilinear 
form $a_L(\cdot,\cdot)$ from Section \ref{sec:continuous_model}, all the 
eigenvalues of $L$ are positive. By item 2 from Proposition \ref{eigenval}, 
all the eigenvalues of $L_h$ are positive as well. This means $L_h$ is a 
positive-definite operator. Since $p$ is a negative real number, $L_h - 
pI_{V_h}$ is a positive-definite operator. Further, $\hat{L}$ is symmetric 
and product of positive-definite matrices, thus by 
\cite[Corollary 7.6.2]{horn2012matrix}, $\hat{L}$ is also positive-definite. 
Therefore, $\hat{L}^{1/2}$ is positive-definite, and since $V_h$ is a finite dimensional space, $\hat{L}^{1/2}: V_h\to V_h$ is an isomorphism. This means that  $V_h = span\{\hat{L}^{1/2}\varphi_j \}_{j = 1}^{n_h}$. Hence, the weak form of \eqref{twostep1} can be written as:
\begin{equation}
	\label{twostep1weakform}
	r^{-1/2}\sum_{j = 1}^{n_h} z_j( \hat{L}^{1/2}\varphi_j, \hat{L}^{1/2}\varphi_l ) _{L_2(\mathcal{D})}  = (\mathcal{W}_h,  \hat{L}^{1/2}\varphi_l )_{L_2(\mathcal{D})},\quad l = 1,...,n_h.
\end{equation}
Define $\widehat{\boldsymbol{L}} = \boldsymbol{L}\boldsymbol{C}^{-1}\boldsymbol{L}-p\boldsymbol{L}$. Then,
by the identity $\boldsymbol{z} = (\boldsymbol{C}^{-1}\boldsymbol{L})^n \boldsymbol{x}$, the self-adjointness of $\hat{L}^{1/2}$ and $(  L_h^2 \varphi_j, \varphi_l ) _{L_2(\mathcal{D})} = [\boldsymbol{L}\boldsymbol{C}^{-1}\boldsymbol{L}]_{j,l} $, the sum in the left hand side of \eqref{twostep1weakform} is
\begin{align}
	\sum_{j = 1}^{n_h} z_j( \hat{L}^{1/2}\varphi_j, \hat{L}^{1/2}\varphi_l ) _{L_2(\mathcal{D})}
	&= 
	\sum_{j,k = 1}^{n_h}  [(\boldsymbol{C}^{-1}\boldsymbol{L})^n]_{j,k} x_k   ( \hat{L}\varphi_j, \varphi_l ) _{L_2(\mathcal{D})} 
	=
	\sum_{j,k = 1}^{n_h} [(\boldsymbol{C}^{-1}\boldsymbol{L})^n]_{j,k} x_k \widehat{\boldsymbol{L}}_{j,l} \nonumber\\
	&= 
	\sum_{k = 1}^{n_h} x_k \sum_{j = 1}^{n_h} \widehat{\boldsymbol{L}}_{l,j} [(\boldsymbol{C}^{-1}\boldsymbol{L})^n]_{j,k} 
	= 
	\label{lhs}
	\sum_{k = 1}^{n_h}  [ \widehat{\boldsymbol{L}} (\boldsymbol{C}^{-1}\boldsymbol{L})^n  ]_{l,k}  x_k. 
\end{align}
Let $\boldsymbol{W} = [(\mathcal{W}_h,  \hat{L}^{1/2}\varphi_1 )_{L_2(\mathcal{D})},...,(\mathcal{W}_h,  \hat{L}^{1/2}\varphi_{n_h} )_{L_2(\mathcal{D})}]^\top$. Since $\mathcal{W}_h$ is white noise in $V_h$, we have $\boldsymbol{W} \sim N(\boldsymbol{0},\widehat{\boldsymbol{L}})$. By \eqref{twostep1weakform} and \eqref{lhs}, $\boldsymbol{x} =  r^{1/2}(\boldsymbol{L}^{-1}\boldsymbol{C})^n \widehat{\boldsymbol{L}}^{-1} \boldsymbol{W} $. Thus, the covariance matrix of $\boldsymbol{x}$ is $r(\boldsymbol{L}^{-1}\boldsymbol{C})^n \widehat{\boldsymbol{L}}^{-1} (\boldsymbol{C}\boldsymbol{L}^{-1})^n$, which also can be written as $r(\boldsymbol{L}^{-1}\boldsymbol{C})^{\lfloor 2\beta \rfloor  }(\boldsymbol{L}-p\boldsymbol{C})^{-1}$. Therefore, $\boldsymbol{x} \sim N(\boldsymbol{0},r(\boldsymbol{L}^{-1}\boldsymbol{C})^{\lfloor 2\beta \rfloor }(\boldsymbol{L}-p\boldsymbol{C})^{-1})$.

If $\lfloor 2\beta \rfloor$ is even, say $\lfloor 2\beta \rfloor = 2n$, with $n$ a non-negative integer ($\lfloor 2\beta \rfloor$ can be 0), we can write the SPDE as
$r^{-1/2}(L_h-p I_{V_h})^{1/2} L_h^n x_h = \mathcal{W}_h.$
In fact, this is a subcase of the previous case. One can simply change the 
$\hat{L}$ to $(L_h-p I_{V_h})$ and the procedure follows similarly. The distribution of $\boldsymbol{x}$ in the case is still $\boldsymbol{x} \sim N(\boldsymbol{0},r(\boldsymbol{L}^{-1}\boldsymbol{C})^{\lfloor 2\beta \rfloor}(\boldsymbol{L}-p\boldsymbol{C})^{-1})$.  

For the second term in \eqref{cov_matrix_rational2}, $k  L_h^{-\lfloor 2\beta \rfloor}$, the corresponding SPDE is
$k^{-1/2}L_h^{-\lfloor 2\beta \rfloor/2} x_h = \mathcal{W}_h.$
Considering again the two cases when $\lfloor 2\beta \rfloor$ is odd or even separately, the derivation follows similarly as above. In both of these cases,  $\boldsymbol{x} \sim N(\boldsymbol{0},k(\boldsymbol{L}^{-1}\boldsymbol{C})^{\lfloor 2\beta \rfloor -1}  \boldsymbol{L}^{-1} )$.
To conclude, observe that we obtained the distribution of each $\boldsymbol{x}_i$ in \eqref{stoc_weights_sum_representaion} for $i = 1,..,m+1$. Therefore, this proves \eqref{cov_matrix_rational2}.

\section{Finite element basis functions}\label{sec:FEM_basis_func}
In this section, we provide explicit forms of the continuous piecewise linear finite element basis functions $\{\varphi_j\}_{j=1}^{n_h}$ mentioned in Sections \ref{finite_element} and \ref{GMRF}.
First, we divide the computational domain $\mathcal{D}$ with a triangle mesh and we call each small triangle an element.
Second, we associate each mesh node to a piecewise linear and continuous basis function. 
The basis function  takes the value $1$ at that node, decreases linearly to the value $0$ at all the neighboring mesh nodes and takes the value $0$ constantly elsewhere on the domain. 

In practice, the basis functions are first defined for a reference element and then mapped to a physical element on mesh. For example,
we can consider a one dimensional domain $\mathcal{D}$,
say $\mathcal{D} = [a,b] \in \mathbb{R}$.
Let 
$ a = x_1 < \cdot\cdot\cdot <x_{n_h} = b$
be a partition of the domain. 
Each sub--interval $[x_i,x_{i+1}],\ i = 1,...,n_h-1$ is an element. We can choose the reference element as $[0,1]$ and define two basis functions on this element as $\varphi_{r,1}(X) = 1-X$ and $\varphi_{r,2}(X) = X$ for $X\in[0,1]$.
Through a change of variables $x = x_i + (x_{i+1}-x_i)X$ for $x\in [x_i,x_{i+1}], i = 1,...,n_h-1$, we can find the basis functions defined on the mesh. For the interior nodes, we have 
$$
\varphi_i(x)=\left \{
\begin{array}{lcl}
	\frac{x - x_{i-1}}{x_i-x_{i-1}},\ x_{i-1}\leq x \leq x_i,\\
	\frac{x_{i+1} - x}{x_{i+1} - x_i},\ x_i\leq x \leq x_{i+1},\ i = 2,...,n_h-1,\\
	0,\ \text{otherwise},
\end{array} \right. 
$$
and for the two boundary nodes we have 
$$
	\varphi_1(x)=\left \{
	\begin{array}{lcl}
		\frac{x_2 - x}{x_2-x_1},\ x_{1}\leq x \leq x_2,\\
		0,\ \text{otherwise},
	\end{array} \right. \quad
	\varphi_{n_h}(x)=\left \{
	\begin{array}{lcl}
		\frac{x - x_{n_h-1}}{x_{n_h}-x_{n_h-1}},\ x_{n_h-1}\leq x \leq x_{n_h},\\
		0,\ \text{otherwise}.
	\end{array} \right. 
$$

This type of basis functions is also referred to as hat functions. The basis functions in higher dimensional spaces generalize naturally from the one dimensional case. Figure~\ref{fig:FEM_2d_basis_func} shows an illustration of a basis function on a two dimensional domain.
In the case of Dirichlet boundary conditions, we remove all basis functions centered at mesh nodes on the boundary, so that the FEM approximation satisfies the Dirichlet boundary conditions.
\begin{figure}[t]
	\centering
	\includegraphics[width=0.6\textwidth]{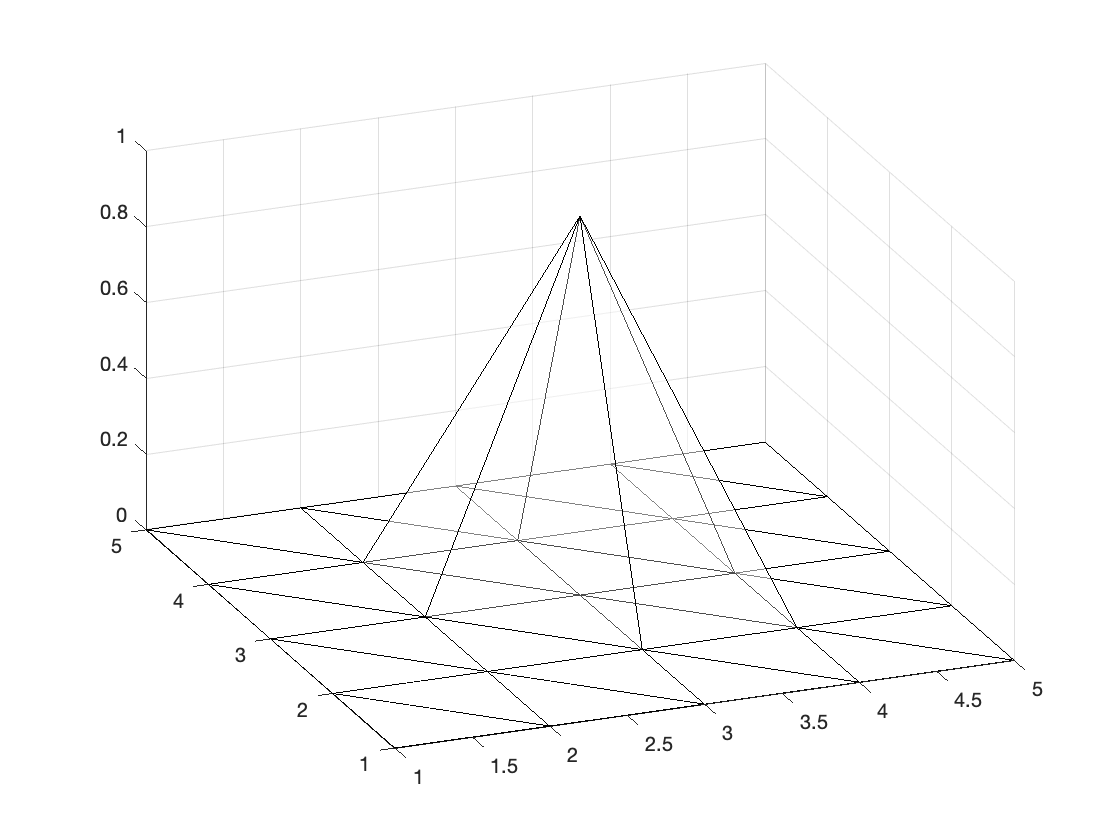}
	\caption{
		Example of a finite element basis function on a mesh in two dimensions. 
	}
	\label{fig:FEM_2d_basis_func}
\end{figure}

\section{Likelihood evaluation and posterior sampling}\label{sec:algorithm}
It is computationally efficient to evaluate the likelihood in \eqref{marginal_likelihood} since $\boldsymbol{Q}$, $\boldsymbol{Q}_{\boldsymbol{X}\vert\boldsymbol{y}}$ and $\boldsymbol{Q}_{\pmb{\epsilon}}$ are sparse.
Specifically, to compute the marginal likelihood \eqref{marginal_likelihood}, Algorithm~\ref{alg1} can be used. 

\begin{algorithm}[h]
	\renewcommand{\algorithmicrequire}{\textbf{Input:}}
	\renewcommand{\algorithmicensure}{\textbf{Output:}}
	\caption{Marginal likelihood computation}
	\label{alg1}
	\begin{algorithmic}[1]
		\STATE Assemble $\overline{\boldsymbol{A}}$, $\boldsymbol{Q}$ and $\boldsymbol{Q}_{\pmb{\epsilon}}$.
		\STATE Compute $\boldsymbol{Q}_{\boldsymbol{X}\vert\boldsymbol{y}}=\overline{\boldsymbol{A}}^\top\boldsymbol{Q}_{\pmb{\epsilon}}\overline{\boldsymbol{A}}+\boldsymbol{Q}$ and 
		$	\pmb{\mu}_{\boldsymbol{X}\vert\boldsymbol{y}}=\boldsymbol{Q}^{-1}_{\boldsymbol{X}\vert\boldsymbol{y}}\overline{\boldsymbol{A}}^\top\boldsymbol{Q}_{\pmb{\epsilon}}\boldsymbol{y}$, where  $\pmb{\mu}_{\boldsymbol{X}\vert\boldsymbol{y}}$ is computed by solving $\boldsymbol{Q}_{\boldsymbol{X}\vert\boldsymbol{y}}\pmb{\mu}_{\boldsymbol{X}\vert\boldsymbol{y}}=\overline{\boldsymbol{A}}^\top\boldsymbol{Q}_{\pmb{\epsilon}}\boldsymbol{y}$ for $\pmb{\mu}_{\boldsymbol{X}\vert\boldsymbol{y}}$ with 
		\citet[][Algorithm 2.1]{rue2005gaussian}.
		\STATE Compute $\log{|\boldsymbol{Q}|}$ by exploiting sparsity of $\boldsymbol{Q}$.
		First, compute the Cholesky decomposition of $\boldsymbol{Q}$: $\boldsymbol{Q} = \boldsymbol{L}\boldsymbol{L}^\top$.
		Then, compute $\log{|\boldsymbol{Q}|} = \log{|\boldsymbol{L}\boldsymbol{L}^\top|} = 2\log{|\boldsymbol{L}|} = 2\sum_i \log{\boldsymbol{L}_{ii}}$
		where $\boldsymbol{L}_{ii}$ denotes $i$th diagonal element of $\boldsymbol{L}$.
		\STATE Compute $\log{|\boldsymbol{Q}_{\boldsymbol{X}\vert\boldsymbol{y}}|}$ and $\log{\boldsymbol{Q}_{\pmb{\epsilon}}}$ in the same way as Step 3.
		\STATE Compute the likelihood of $\boldsymbol{y}$ by using \eqref{marginal_likelihood}
	\end{algorithmic}  
\end{algorithm}
Similarly, samples from predictive distributions of the latent field can be obtained effectively via Algorithm~\ref{alg2}.

\begin{algorithm}[h]
	\renewcommand{\algorithmicrequire}{\textbf{Input:}}
	\renewcommand{\algorithmicensure}{\textbf{Output:}}
	\caption{Predictive distribution sampling}
	\label{alg2} 
	\begin{algorithmic}[1]
		\REQUIRE Locations $\textbf{s}_1,\cdots,\textbf{s}_N$ where $u(\textbf{s})$ should be sampled.
		\STATE Assemble $\boldsymbol{Q}$ and $\boldsymbol{Q}_{\pmb{\epsilon}}$.
		\STATE Compute $\boldsymbol{Q}_{\boldsymbol{X}\vert\boldsymbol{y}}$ and $\pmb{\mu}_{\boldsymbol{X}\vert\boldsymbol{y}}$ 
		in the same way as Step 2 in Algorithm \ref{alg1}.
		\STATE Use $\boldsymbol{Q}_{\boldsymbol{X}\vert\boldsymbol{y}}$ and $\pmb{\mu}_{\boldsymbol{X}\vert\boldsymbol{y}}$ to sample $\boldsymbol{X}\vert\boldsymbol{y}$
		by following \citet[][Algorithm 2.4]{rue2005gaussian}.
		\STATE Construct a projection matrix $\overline{\boldsymbol{A}}_{\text{new}}$ for the locations $\textbf{s}_1,\cdots,\textbf{s}_N$.
		\STATE Return $\overline{\boldsymbol{A}}_{\text{new}}\boldsymbol{X}\vert\boldsymbol{y}$ as a sample from $\pi(u(\textbf{s}_1),\cdots,u(\textbf{s}_N)\vert\boldsymbol{y})$. 
	\end{algorithmic}  
\end{algorithm}

\section{Ideas of rational approximation algorithms}\label{sec:algo_idea}

In this section, we will briefly introduce the ideas of the BRASIL algorithm and the Clenshaw-Lord Chebyshev-Pad\'e algorithm that were mentioned in Sections~\ref{rationalapprox} and \ref{GMRF}.

The idea of the BRASIL algorithm is that one can achieve the best rational approximation of a continuous function on a compact interval $[a,b] \in \mathbb{R}$
by interpolating a certain number, depending on the degree of the rational function, of points such that the maximum error of the approximation in each sub--interval divided by those points are equal.
The BRASIL algorithm first initializes a partition of the interval $[a,b]$ by a set of points, then uses the barycentric rational interpolation on those points, and adjusts iteratively the partition so that
the maximum absolute errors in each sub--interval are approximately equal. See \citet[][Section 3]{hofreither2021algorithm} for a complete description of the algorithm.

The Clenshaw-Lord Chebyshev-Pad\'e algorithm approximates the target function by a combination of a Pad\'e approximation and a Chebyshev series. Pad\'e approximation consists of approximating a target function $f$ by a rational function $R_{[m/n]}$, with degree $m$ and $n$ for the numerator and denominator polynomials, respectively. The coefficients of the polynomials are computed so that the derivatives at $0$ agree with the derivatives of the target function up to the highest possible order. That is, $f^{(k)}(0) = R^{(k)}_{m/n}(0)$ for $k = 0,...,m+n$. Now, for any continuous function $f$ on interval $[-1,1]\in \mathbb{R}$, there is a unique Chebyshev series,
which has the form $f(x) = \sum_{k = 0}^{\infty}a_kT_k(x)$, that converges uniformly to the function $f$. Here, $\{a_k\}_k$ are called the Chebyshev coefficients and $\{T_k\}_k$ are the Chebyshev polynomial of the first kind. $\{T_k\}_k$ are defined from the recurrence relation: $T_0(x) = 1,$ $T_1(x) = x,$ and $T_{n+1}(x) = 2xT_n(x) - T_{n-1}(x),$ for $x\in[-1,1]$.
$a_k$ can be computed by $a_k = 2/\pi\int_{-1}^{1}\frac{f(x)T_k(x)\text{d}x}{\sqrt{1-x^2}}$.
One can obtain the Chebyshev series of a continuous function on a compact interval $[a,b]$ through a change of variables. To approximate a continuous function on an interval $[a,b]$, the 
Clenshaw-Lord Chebyshev-Pad\'e algorithm first expands a continuous function on $[a,b]$ with its (truncated) Chebyshev series, then uses Pad\'e approximation to approximate the series.

The two algorithms compute the best or the near best coefficients of rational approximation in the sense of $L_\infty$-norm.
The main reason for computing the coefficients in this way is that we by  \cite{Stahl} then have an explicit rate of convergence of the error, which allows us to compute the explicit bounds for the covariance error. If we only had a bound in the $L^2$-norm (say), it would be less clear how to use that in the theoretical analysis. Further, as far as we know, there are no known methods for obtaining optimal rational approximations with respect to other norms that have known rates of convergence. 

There are other methods for computing rational approximations of fractional powers of elliptic operators, based on alternative representations of the fractional power. One example is the method of \citet{bonito_approximation} which was applied in \citet{bolin2020numerical}. That method, however, has a much higher error for low orders of the rational approximation \citep{rational_spde}, and is therefore not suitable in our context. 

\section{Further numerical experiments}\label{sec:plot_likelihooderror}
In this section, we provide some additional details on the numerical experiments and provide some plots of the absolute relative errors of the likelihood shown in \eqref{marginal_likelihood} with different orders of the rational approximations.

First we describe how we approximate the ${L_2([0,1]^2\times [0,1]^2)}$-norm and the supremum norm on $[0,1]^2\times [0,1]^2$. In order to approximate these norms we first need to build some matrices induced by the covariance operators. First, denote by $\{\boldsymbol{s}_i\}_{i = 1}^{N^2}$ the locations of the mesh nodes. 
For two continuous functions $\rho,\hat{\rho}:[0,1]^2\times [0,1]^2 \to \mathbb{R}$, let $\boldsymbol{\Sigma}$ and $\hat{\boldsymbol{\Sigma}}$ be $N^2\times N^2$ matrices with corresponding $(i,j)$th elements given by $\boldsymbol{\Sigma}(i,j) = \rho(\boldsymbol{s}_i,\boldsymbol{s}_j)$ and $\hat{\boldsymbol{\Sigma}}(i,j) = \hat{\rho}(\boldsymbol{s}_i,\boldsymbol{s}_j)$, respectively. The $L_2([0,1]^2\times [0,1]^2)$-distance between $\rho$ and $\hat{\rho}$ can be can be approximated, on this regular mesh, by the following quadrature:
\begin{equation}\label{eq:L2DDnorm}
	\|\rho-\hat{\rho}\|_{L_2([0,1]^2\times [0,1]^2)} \approx \sqrt{\frac{1}{N^4} \sum_{i=1}^{N^2} \sum_{j=1}^{N^2} \left(\rho(\boldsymbol{s}_i,\boldsymbol{s}_j)-\hat{\rho}(\boldsymbol{s}_i,\boldsymbol{s}_j)\right)^2} = \frac{1}{N^2} \|\boldsymbol{\Sigma}-\hat{\boldsymbol{\Sigma}}\|_F,
\end{equation}
where $\|\cdot \|_F$ stands for the Frobenius norm. Similarly, we can approximate the supremum distance between $\rho$ and $\hat{\rho}$ by the max-distance on the corresponding matrices:
\begin{equation}\label{eq:maxnormsup}
	\|\rho-\hat{\rho}\|_{L_{\infty}([0,1]^2\times [0,1]^2)} \approx \max_{i,j} |\rho(\boldsymbol{s}_i,\boldsymbol{s}_j) - \hat{\rho}(\boldsymbol{s}_i,\boldsymbol{s}_j)| = \|\boldsymbol{\Sigma}-\hat{\boldsymbol{\Sigma}}\|_{\max},
\end{equation}
where $\|\cdot\|_{\max}$ stands for the max norm.
Thus, to approximate the errors, we just need to assemble the true covariance matrix and the covariance matrix of the approximation. Let us now describe how this is done. 
To this end, fix some smoothness parameter $\nu>0$, and let $\beta = \nu/2+d/4$. We build the covariance matrix $\boldsymbol{\Sigma}^{\beta}$, of size $N^2\times N^2$, associated to the true covariance function by setting its $(i,j)$th element to be $\boldsymbol{\Sigma}^{\beta}_{i,j} = \varrho_u^\beta(\boldsymbol{s}_i, \boldsymbol{s}_j)$, where $\varrho_u^\beta$ is given in \eqref{eq:foldedmatern}. In practice, we truncate the sum in \eqref{eq:foldedmatern} to a sufficiently large range of $\mathbf{k} \in \mathbb{Z}^2$. Let $\boldsymbol{Q}_{I,\beta}$ be the precision matrix obtained from INLA's method of general smoothness, with corresponding covariance matrix  $\boldsymbol{\Sigma}_{I}^{\beta} = \boldsymbol{Q}_{I,\beta}^{-1}$. Now, fix some order $m$ for the rational approximation and let $\boldsymbol{Q}_{m,O,\beta}$ be the precision matrix from the operator-based rational approximation of order $m$. The covariance matrix associated to the operator-based rational approximation is given by $\boldsymbol{\Sigma}_{O,m}^{\beta} = \boldsymbol{Q}_{m,O,\beta}^{-1}$. Finally, let $\boldsymbol{Q}_{m,C,\beta}$ be the precision matrix given by \eqref{Q}. The corresponding covariance matrix is then given by $\boldsymbol{\Sigma}_{C,m}^{\beta} = \overline{\boldsymbol{I}} \boldsymbol{Q}_{m,C,\beta}^{-1} \overline{\boldsymbol{I}}^\top$, where  $\overline{\boldsymbol{I}}$ is a block matrix of size $N^2\times N^2(m+1)$ obtained by combining $m+1$ copies of the $N^2\times N^2$ identity matrix $\boldsymbol{I}_{N^2}$ as  $\overline{\boldsymbol{I}}=\begin{bmatrix}\boldsymbol{I}_{N^2}&\cdots&\boldsymbol{I}_{N^2}\end{bmatrix}.$ 
\begin{figure}[t]
	\centering 
	\includegraphics[width=0.865\textwidth]{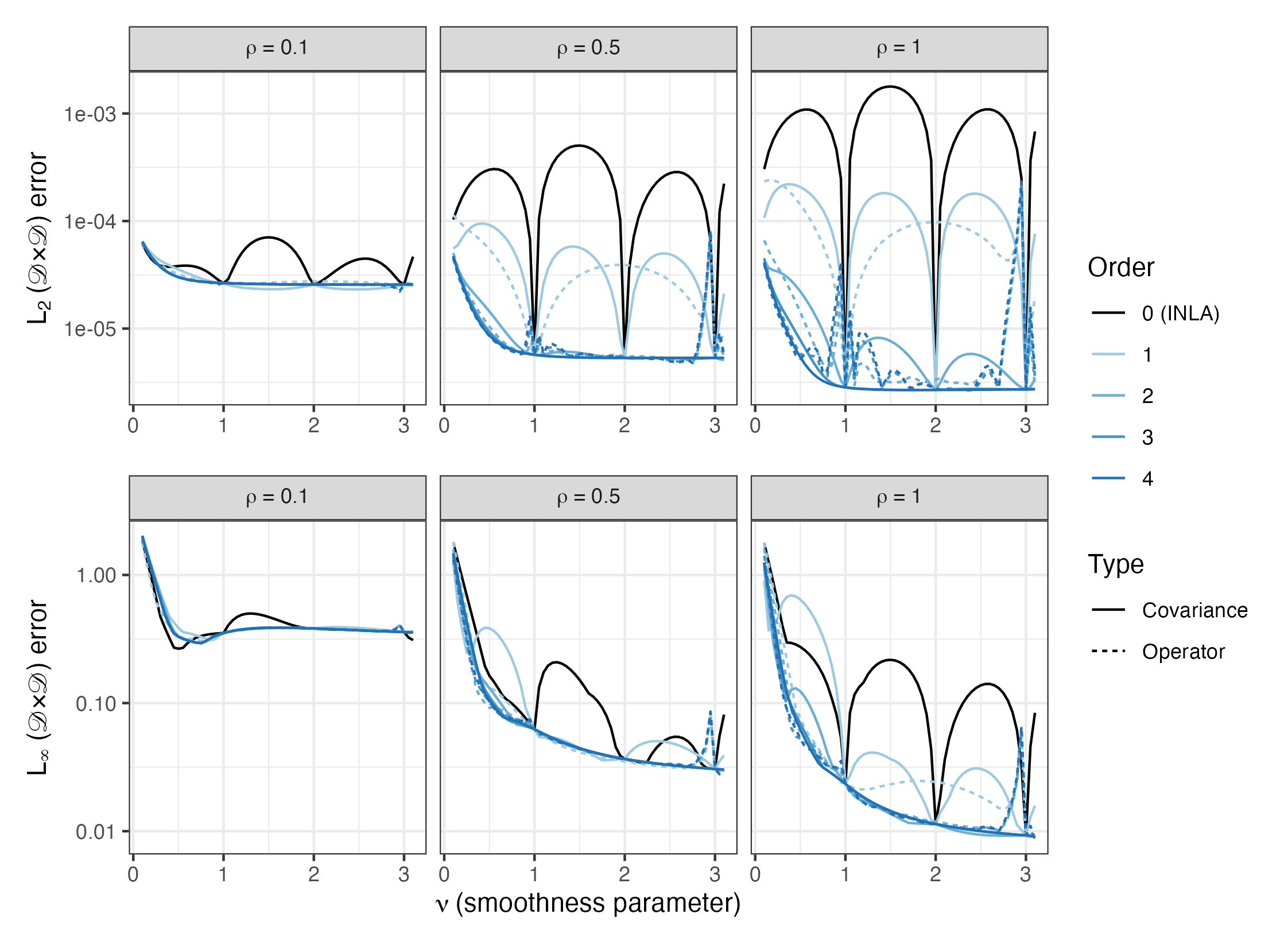}
	\caption{Errors in $L_2(\mathcal{D}\times \mathcal{D})$-norm (top) and supremum norm ($L_{\infty}(\mathcal{D}\times \mathcal{D})$) (bottom) on $\mathcal{D} = [0,1]^2$ for different practical ranges $\rho$ for different values of $\nu$. All methods use the same FEM mesh, with $50$ equally spaced nodes in each direction.}
	\label{fig:L2normCov50}
\end{figure}

The results of the covariance error for the coarser FEM with $50$ equally spaced nodes on each axis can be seen in Figure~\ref{fig:L2normCov50}.
We now consider similar a comparison for the likelihood errors of the different methods. 
For the comparison, we generate 1000 sets of samples, where each contains 1000 observations on $\mathcal{D} = [0,1]^2$ generated from \eqref{hiermodel} where $u$  has covariance function \eqref{eq:foldedmatern}.
For each set of samples $\mathbf{y}_i$, we compute the true log-likelihood value $\ell(\mathbf{y}_i)$ and the approximation $\hat{\ell}(\mathbf{y}_i)$ for each of the three methods, and finally store the absolute relative error $|1-\hat{\ell}(\mathbf{y}_i)/\ell(\mathbf{y}_i)|$. The median of the 1000 absolute relative errors for the three methods are presented in log scale in  Figure~\ref{fig:likelihood_error}.
\begin{figure}[t]
	\centering
	\includegraphics[width=0.865\textwidth]{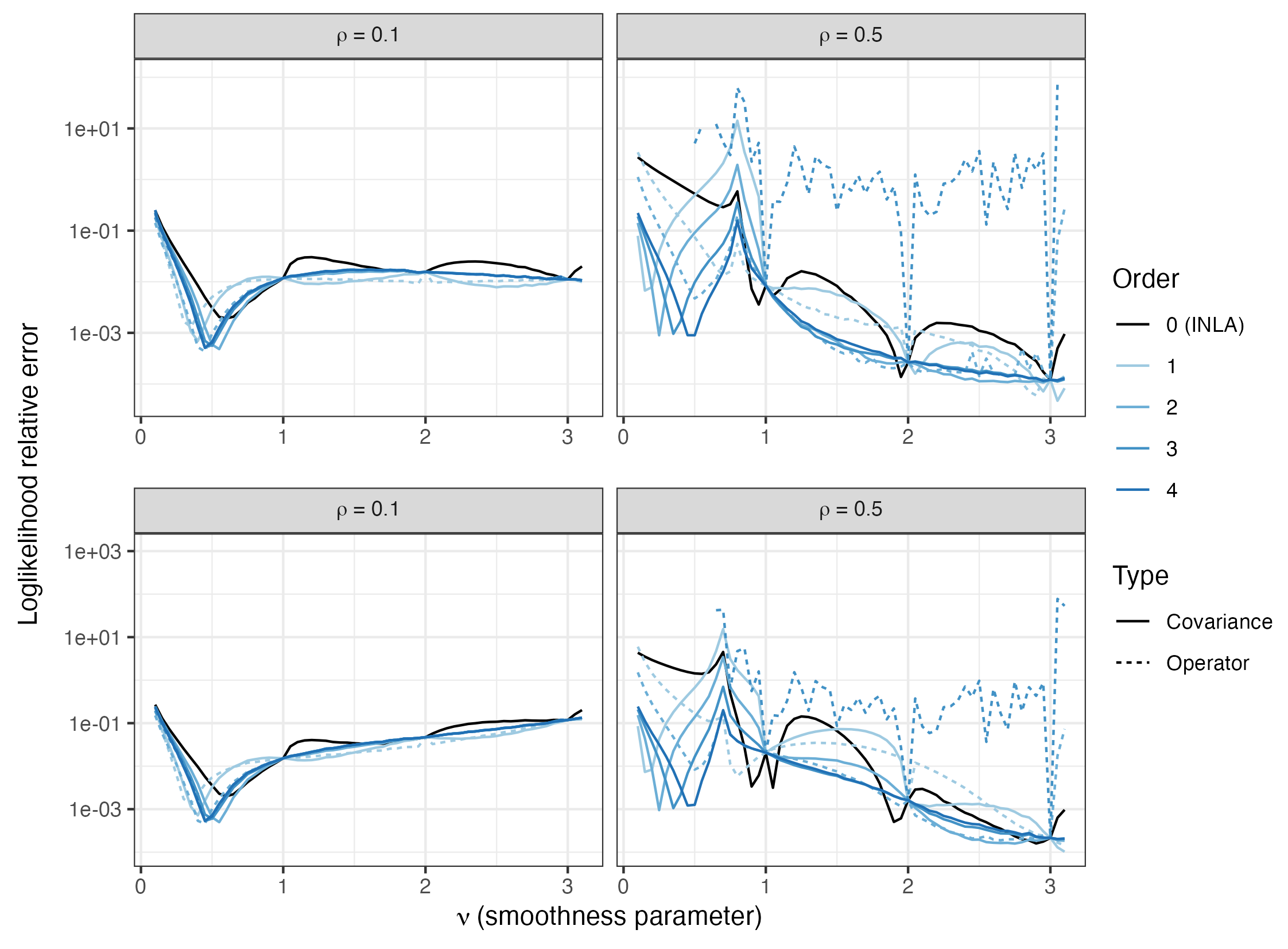}
	\caption{
		The log-scaled log-likelihood errors, where $\rho$ is the range parameter and the standard deviation of the measurement noise is 0.1 (top) and 0.01 (bottom). All the methods use the same FEM mesh, with $100$ equally spaced nodes in each direction.
	}
	\label{fig:likelihood_error}
\end{figure}

We can note that the error tends to decrease when the order of the rational approximation, $m$, increases.
Recall that the error for $\nu\in\mathbb{N}$ solely comes from the FEM error, so we can see that there is no need for a large $m$ to obtain an error which is on the same scale as the FEM error. 
In fact, the likelihood error for integer $\nu$ and non-integer $\nu$ are quite similar as long as $m\geq 2$. This means that we essentially have the same likelihood error for a general $\nu$ with our method as the standard SPDE approach has for integer values of $\nu$ (where the error only comes from the FEM discretization).
Finally, we can also note that the covariance--based method has better numerical stability with respect to $m$ compared with the operator--based method. 
More comparisons can be found the \texttt{Shiny} app at \url{https://github.com/davidbolin/rSPDE}.
\end{appendix}

\section*{Acknowledgement}
Our sincere thanks to Elias T. Krainski and H{\aa}vard Rue for their help with explaining some details of the internal structure of the \texttt{R-INLA} software
and to the anonymous reviewers for insightful comments and suggestions on the article.

\bibliographystyle{chicago}
\bibliography{reference_list}

\begin{thebibliography}{}

\bibitem[\protect\citeauthoryear{Bachl, Lindgren, Borchers, and Illian}{Bachl
  et~al.}{2019}]{inlabru}
Bachl, F.~E., F.~Lindgren, D.~L. Borchers, and J.~B. Illian (2019).
\newblock {inlabru}: an {R} package for bayesian spatial modelling from
  ecological survey data.
\newblock {\em Methods in Ecology and Evolution\/}~{\em 10}, 760--766.

\bibitem[\protect\citeauthoryear{Baker and Graves-Morris}{Baker and
  Graves-Morris}{1996}]{baker1996pade}
Baker, Jr., G.~A. and P.~Graves-Morris (1996).
\newblock {\em Pad\'{e} approximants\/} (Second ed.), Volume~59 of {\em
  Encyclopedia of Mathematics and its Applications}.
\newblock Cambridge University Press, Cambridge.

\bibitem[\protect\citeauthoryear{Banerjee, Carlin, and Gelfand}{Banerjee
  et~al.}{2015}]{banerjee2003hierarchical}
Banerjee, S., B.~P. Carlin, and A.~E. Gelfand (2015).
\newblock {\em Hierarchical modeling and analysis for spatial data\/} (Second
  ed.), Volume 135 of {\em Monographs on Statistics and Applied Probability}.
\newblock CRC Press, Boca Raton, FL.

\bibitem[\protect\citeauthoryear{Bolin and Kirchner}{Bolin and
  Kirchner}{2020}]{rational_spde}
Bolin, D. and K.~Kirchner (2020).
\newblock The rational {SPDE} approach for {G}aussian random fields with
  general smoothness.
\newblock {\em J. Comput. Graph. Statist.\/}~{\em 29\/}(2), 274--285.

\bibitem[\protect\citeauthoryear{Bolin and Kirchner}{Bolin and
  Kirchner}{2023}]{bk2022measure}
Bolin, D. and K.~Kirchner (2023).
\newblock Equivalence of measures and asymptotically optimal linear prediction
  for {G}aussian random fields with fractional-order covariance operators.
\newblock {\em Bernoulli\/}~{\em 29\/}(2), 1476--1504.

\bibitem[\protect\citeauthoryear{Bolin, Kirchner, and Kov\'{a}cs}{Bolin
  et~al.}{2018}]{bolin2018weak}
Bolin, D., K.~Kirchner, and M.~Kov\'{a}cs (2018).
\newblock Weak convergence of {G}alerkin approximations for fractional elliptic
  stochastic {PDE}s with spatial white noise.
\newblock {\em BIT\/}~{\em 58\/}(4), 881--906.

\bibitem[\protect\citeauthoryear{Bolin, Kirchner, and Kov\'{a}cs}{Bolin
  et~al.}{2020}]{bolin2020numerical}
Bolin, D., K.~Kirchner, and M.~Kov\'{a}cs (2020).
\newblock Numerical solution of fractional elliptic stochastic {PDE}s with
  spatial white noise.
\newblock {\em IMA J. Numer. Anal.\/}~{\em 40\/}(2), 1051--1073.

\bibitem[\protect\citeauthoryear{Bolin and Lindgren}{Bolin and
  Lindgren}{2011}]{bolin11}
Bolin, D. and F.~Lindgren (2011).
\newblock Spatial models generated by nested stochastic partial differential
  equations, with an application to global ozone mapping.
\newblock {\em Ann. Appl. Stat.\/}~{\em 5\/}(1), 523--550.

\bibitem[\protect\citeauthoryear{Bolin and Simas}{Bolin and
  Simas}{2023}]{rspde}
Bolin, D. and A.~B. Simas (2023).
\newblock {\em rSPDE: Rational Approximations of Fractional Stochastic Partial
  Differential Equations}.
\newblock R package version 2.2.0.

\bibitem[\protect\citeauthoryear{Bolin, Simas, and Wallin}{Bolin
  et~al.}{2022}]{bolinsimaswallin}
Bolin, D., A.~B. Simas, and J.~Wallin (2022).
\newblock Gaussian {W}hittle-{M}atérn fields on metric graphs.
\newblock arXiv:2205.06163.

\bibitem[\protect\citeauthoryear{Bonito and Pasciak}{Bonito and
  Pasciak}{2015}]{bonito_approximation}
Bonito, A. and J.~E. Pasciak (2015).
\newblock Numerical approximation of fractional powers of elliptic operators.
\newblock {\em Math. Comp.\/}~{\em 84\/}(295), 2083--2110.

\bibitem[\protect\citeauthoryear{Chandler-Wilde, Hewett, and
  Moiola}{Chandler-Wilde et~al.}{2015}]{chandlerwildeetal}
Chandler-Wilde, S.~N., D.~P. Hewett, and A.~Moiola (2015).
\newblock Interpolation of {H}ilbert and {S}obolev spaces: quantitative
  estimates and counterexamples.
\newblock {\em Mathematika\/}~{\em 61\/}(2), 414--443.

\bibitem[\protect\citeauthoryear{Chang, Cheng, Allaire, Sievert, Schloerke,
  Xie, Allen, McPherson, Dipert, and Borges}{Chang et~al.}{2021}]{shiny}
Chang, W., J.~Cheng, J.~Allaire, C.~Sievert, B.~Schloerke, Y.~Xie, J.~Allen,
  J.~McPherson, A.~Dipert, and B.~Borges (2021).
\newblock {\em shiny: Web Application Framework for R}.
\newblock R package version 1.6.0.

\bibitem[\protect\citeauthoryear{Ciarlet}{Ciarlet}{2002}]{ciarletfinitelement}
Ciarlet, P.~G. (2002).
\newblock {\em The finite element method for elliptic problems}, Volume~40 of
  {\em Classics in Applied Mathematics}.
\newblock Society for Industrial and Applied Mathematics (SIAM), Philadelphia,
  PA.
\newblock Reprint of the 1978 original [North-Holland, Amsterdam].

\bibitem[\protect\citeauthoryear{Cox and Kirchner}{Cox and
  Kirchner}{2020}]{cox_Kirchner}
Cox, S.~G. and K.~Kirchner (2020).
\newblock Regularity and convergence analysis in {S}obolev and {H}\"{o}lder
  spaces for generalized {W}hittle-{M}at\'{e}rn fields.
\newblock {\em Numer. Math.\/}~{\em 146\/}(4), 819--873.

\bibitem[\protect\citeauthoryear{Davies}{Davies}{1995}]{davies1996spectral}
Davies, E.~B. (1995).
\newblock {\em Spectral Theory and Differential Operators}.
\newblock Cambridge Studies in Advanced Mathematics. Cambridge University
  Press.

\bibitem[\protect\citeauthoryear{Evans and Gariepy}{Evans and
  Gariepy}{2015}]{evansfineproperties}
Evans, L.~C. and R.~F. Gariepy (2015).
\newblock {\em Measure theory and fine properties of functions\/} (Revised
  ed.).
\newblock Textbooks in Mathematics. CRC Press, Boca Raton, FL.

\bibitem[\protect\citeauthoryear{Fedosov}{Fedosov}{1963}]{fedosov1}
Fedosov, B. (1963).
\newblock Asymptotic formulas for the eigenvalues of the laplacian in the case
  of a polygonal region.
\newblock {\em Sov. Math., Dokl.\/}~{\em 4}, 1092--1096.

\bibitem[\protect\citeauthoryear{Fedosov}{Fedosov}{1964}]{fedosov2}
Fedosov, B. (1964).
\newblock Asymptotic formulas for the eigenvalues of the laplace operator in
  the case of a polyhedron.
\newblock {\em Sov. Math., Dokl.\/}~{\em 5}, 988--990.

\bibitem[\protect\citeauthoryear{Fuglstad, Simpson, Lindgren, and Rue}{Fuglstad
  et~al.}{2015}]{fuglstad2015does}
Fuglstad, G.-A., D.~Simpson, F.~Lindgren, and H.~Rue (2015).
\newblock Does non-stationary spatial data always require non-stationary random
  fields?
\newblock {\em Spat. Stat.\/}~{\em 14\/}(part B), 505--531.

\bibitem[\protect\citeauthoryear{Fuglstad, Simpson, Lindgren, and Rue}{Fuglstad
  et~al.}{2019}]{fuglstad2019constructing}
Fuglstad, G.-A., D.~Simpson, F.~Lindgren, and H.~Rue (2019).
\newblock Constructing priors that penalize the complexity of {G}aussian random
  fields.
\newblock {\em J. Amer. Statist. Assoc.\/}~{\em 114\/}(525), 445--452.

\bibitem[\protect\citeauthoryear{Good}{Good}{1952}]{logscore}
Good, I.~J. (1952).
\newblock Rational decisions.
\newblock {\em Journal of the Royal Statistical Society: Series B
  (Methodological)\/}~{\em 14\/}(1), 107--114.

\bibitem[\protect\citeauthoryear{Grisvard}{Grisvard}{1967}]{grisvardinterpolation}
Grisvard, P. (1967).
\newblock Caract\'{e}risation de quelques espaces d'interpolation.
\newblock {\em Arch. Rational Mech. Anal.\/}~{\em 25}, 40--63.

\bibitem[\protect\citeauthoryear{Grisvard}{Grisvard}{2011}]{grisvard}
Grisvard, P. (2011).
\newblock {\em Elliptic problems in nonsmooth domains}, Volume~69 of {\em
  Classics in Applied Mathematics}.
\newblock Society for Industrial and Applied Mathematics (SIAM), Philadelphia,
  PA.

\bibitem[\protect\citeauthoryear{Heaton, Datta, Finley, and et~al.}{Heaton
  et~al.}{2019}]{heaton2019}
Heaton, M.~J., A.~Datta, A.~O. Finley, and et~al. (2019).
\newblock A case study competition among methods for analyzing large spatial
  data.
\newblock {\em J. Agric. Biol. Environ. Stat.\/}~{\em 24\/}(3), 398--425.

\bibitem[\protect\citeauthoryear{Herrmann, Kirchner, and Schwab}{Herrmann
  et~al.}{2020}]{herrmann2020multilevel}
Herrmann, L., K.~Kirchner, and C.~Schwab (2020).
\newblock Multilevel approximation of {G}aussian random fields: fast
  simulation.
\newblock {\em Math. Models Methods Appl. Sci.\/}~{\em 30\/}(1), 181--223.

\bibitem[\protect\citeauthoryear{Hildeman, Bolin, and Rychlik}{Hildeman
  et~al.}{2021}]{Hildeman2020}
Hildeman, A., D.~Bolin, and I.~Rychlik (2021).
\newblock Deformed {SPDE} models with an application to spatial modeling of
  significant wave height.
\newblock {\em Spat. Stat.\/}~{\em 42}, Paper No. 100449, 27.

\bibitem[\protect\citeauthoryear{Hofreither}{Hofreither}{2021}]{hofreither2021algorithm}
Hofreither, C. (2021).
\newblock An algorithm for best rational approximation based on barycentric
  rational interpolation.
\newblock {\em Numer. Algorithms\/}~{\em 88\/}(1), 365--388.

\bibitem[\protect\citeauthoryear{Horn and Johnson}{Horn and
  Johnson}{2013}]{horn2012matrix}
Horn, R.~A. and C.~R. Johnson (2013).
\newblock {\em Matrix analysis\/} (Second ed.).
\newblock Cambridge University Press, Cambridge.

\bibitem[\protect\citeauthoryear{Kaufman, Schervish, and Nychka}{Kaufman
  et~al.}{2008}]{cov_taper}
Kaufman, C.~G., M.~J. Schervish, and D.~W. Nychka (2008).
\newblock Covariance tapering for likelihood-based estimation in large spatial
  data sets.
\newblock {\em J. Amer. Statist. Assoc.\/}~{\em 103\/}(484), 1545--1555.

\bibitem[\protect\citeauthoryear{Khristenko, Scarabosio, Swierczynski, Ullmann,
  and Wohlmuth}{Khristenko et~al.}{2019}]{khristenkoetal}
Khristenko, U., L.~Scarabosio, P.~Swierczynski, E.~Ullmann, and B.~Wohlmuth
  (2019).
\newblock Analysis of boundary effects on {PDE}-based sampling of
  {W}hittle-{M}at\'{e}rn random fields.
\newblock {\em SIAM/ASA J. Uncertain. Quantif.\/}~{\em 7\/}(3), 948--974.

\bibitem[\protect\citeauthoryear{Kirszbraun}{Kirszbraun}{1934}]{kirszbraun}
Kirszbraun, M.~D. (1934).
\newblock Über die zusammenziehende und lipschitzsche transformationen.
\newblock {\em Fund. Math.\/}~{\em 22}, 77--108.

\bibitem[\protect\citeauthoryear{Lindgren, Bakka, Bolin, Krainski, and
  Rue}{Lindgren et~al.}{2020}]{bakka_diffusion-based_2020}
Lindgren, F., H.~Bakka, D.~Bolin, E.~Krainski, and H.~Rue (2020).
\newblock A diffusion-based spatio-temporal extension of {G}aussian
  {M}at\'{e}rn fields.
\newblock arXiv: 2006.04917v2.

\bibitem[\protect\citeauthoryear{Lindgren, Bolin, and Rue}{Lindgren
  et~al.}{2022}]{lindgren2022spde}
Lindgren, F., D.~Bolin, and H.~Rue (2022).
\newblock The {SPDE} approach for {G}aussian and non-{G}aussian fields: 10
  years and still running.
\newblock {\em Spat. Stat.\/}~{\em 50}, Paper No. 100599.

\bibitem[\protect\citeauthoryear{Lindgren and Rue}{Lindgren and
  Rue}{2015}]{lindgren2015software}
Lindgren, F. and H.~Rue (2015).
\newblock Bayesian spatial modelling with {R-INLA}.
\newblock {\em Journal of Statistical Software\/}~{\em 63\/}(19), 1–25.

\bibitem[\protect\citeauthoryear{Lindgren, Rue, and Lindstr\"{o}m}{Lindgren
  et~al.}{2011}]{lindgren2011}
Lindgren, F., H.~Rue, and J.~Lindstr\"{o}m (2011).
\newblock An explicit link between {G}aussian fields and {G}aussian {M}arkov
  random fields: the stochastic partial differential equation approach.
\newblock {\em J. R. Stat. Soc. Ser. B Stat. Methodol.\/}~{\em 73\/}(4),
  423--498.

\bibitem[\protect\citeauthoryear{Liu and Rue}{Liu and Rue}{2022}]{lgocv}
Liu, Z. and H.~Rue (2022).
\newblock Leave-group-out cross-validation for latent {G}aussian models.
\newblock arXiv:2210.04482.

\bibitem[\protect\citeauthoryear{Lototsky and Rozovsky}{Lototsky and
  Rozovsky}{2017}]{SPDEbrownucla}
Lototsky, S.~V. and B.~L. Rozovsky (2017).
\newblock {\em Stochastic partial differential equations}.
\newblock Universitext. Springer, Cham.

\bibitem[\protect\citeauthoryear{Mat\'{e}rn}{Mat\'{e}rn}{1960}]{matern60}
Mat\'{e}rn, B. (1960).
\newblock {\em Spatial variation: {S}tochastic models and their application to
  some problems in forest surveys and other sampling investigations}.
\newblock Statens Skogsforskningsinstitut, Stockholm.
\newblock Meddelanden Fr{\aa}n Statens Skogsforskningsinstitut, Band 49, Nr. 5.

\bibitem[\protect\citeauthoryear{{R Core Team}}{{R Core
  Team}}{2022}]{Rsoftware}
{R Core Team} (2022).
\newblock {\em R: A Language and Environment for Statistical Computing}.
\newblock Vienna, Austria: R Foundation for Statistical Computing.

\bibitem[\protect\citeauthoryear{Remez}{Remez}{1934}]{remez1934determination}
Remez, E.~Y. (1934).
\newblock Sur la d{\'e}termination des polyn{\^o}mes d’approximation de
  degr{\'e} donn{\'e}e.
\newblock {\em Comm. Soc. Math. Kharkov\/}~{\em 10\/}(196), 41--63.

\bibitem[\protect\citeauthoryear{Rue and Held}{Rue and
  Held}{2005}]{rue2005gaussian}
Rue, H. and L.~Held (2005).
\newblock {\em Gaussian {M}arkov random fields}, Volume 104 of {\em Monographs
  on Statistics and Applied Probability}.
\newblock Chapman \& Hall/CRC, Boca Raton, FL.
\newblock Theory and applications.

\bibitem[\protect\citeauthoryear{Rue, Martino, and Chopin}{Rue
  et~al.}{2009}]{rue09}
Rue, H., S.~Martino, and N.~Chopin (2009).
\newblock Approximate {B}ayesian inference for latent {G}aussian models by
  using integrated nested {L}aplace approximations.
\newblock {\em J. R. Stat. Soc. Ser. B Stat. Methodol.\/}~{\em 71\/}(2),
  319--392.

\bibitem[\protect\citeauthoryear{Simpson, Rue, Riebler, Martins, and
  S{\o}rbye}{Simpson et~al.}{2017}]{simpson2017penalising}
Simpson, D., H.~Rue, A.~Riebler, T.~G. Martins, and S.~H. S{\o}rbye (2017).
\newblock Penalising model component complexity: a principled, practical
  approach to constructing priors.
\newblock {\em Statist. Sci.\/}~{\em 32\/}(1), 1--28.

\bibitem[\protect\citeauthoryear{Stahl}{Stahl}{2003}]{Stahl}
Stahl, H.~R. (2003).
\newblock Best uniform rational approximation of {$x^\alpha$} on {$[0,1]$}.
\newblock {\em Acta Math.\/}~{\em 190\/}(2), 241--306.

\bibitem[\protect\citeauthoryear{Stein}{Stein}{1999}]{stein1999interpolation}
Stein, M.~L. (1999).
\newblock {\em Interpolation of spatial data}.
\newblock Springer Series in Statistics. Springer-Verlag, New York.
\newblock Some theory for Kriging.

\bibitem[\protect\citeauthoryear{Stein}{Stein}{2002}]{screening_kriging}
Stein, M.~L. (2002).
\newblock {The screening effect in Kriging}.
\newblock {\em Ann. Statist.\/}~{\em 30\/}(1), 298 -- 323.

\bibitem[\protect\citeauthoryear{Steinwart and Scovel}{Steinwart and
  Scovel}{2012}]{kernel_mercer}
Steinwart, I. and C.~Scovel (2012).
\newblock Mercer's theorem on general domains: on the interaction between
  measures, kernels, and {RKHS}s.
\newblock {\em Constr. Approx.\/}~{\em 35\/}(3), 363--417.

\bibitem[\protect\citeauthoryear{Strang and Fix}{Strang and
  Fix}{2008}]{strang_fix}
Strang, G. and G.~Fix (2008).
\newblock {\em An Analysis of the Finite Element Method}.
\newblock Wellesley-Cambridge Press.

\bibitem[\protect\citeauthoryear{Whittle}{Whittle}{1963}]{whittle1963stochastic}
Whittle, P. (1963).
\newblock Stochastic processes in several dimensions.
\newblock {\em Bull. Inst. Internat. Statist.\/}~{\em 40}, 974--994.

\end{thebibliography}

\end{document}